\documentclass[11pt]{article}
\usepackage[T1]{fontenc}
\usepackage[utf8]{inputenc}
\usepackage{verbatim}
\usepackage{mathtools}
\usepackage{amsmath}
\usepackage{amsthm}
\usepackage{amssymb}
\usepackage[bookmarks=true,bookmarksnumbered=false,bookmarksopen=false,
 breaklinks=false,pdfborder={0 0 1},colorlinks=false, pagebackref=false]
 {hyperref}
\hypersetup{
 colorlinks,citecolor=blue,linkcolor=red}

\usepackage{mdframed}

\makeatletter
\theoremstyle{plain}
\newtheorem{thm}{\protect\theoremname}
\theoremstyle{plain}
\newtheorem{cor}[thm]{\protect\corollaryname}
\theoremstyle{plain}
\newtheorem{lem}[thm]{\protect\lemmaname}
\theoremstyle{definition}
\newtheorem{defn}[thm]{\protect\definitionname}
\theoremstyle{plain}
\newtheorem{prop}[thm]{\protect\propositionname}

\@ifundefined{date}{}{\date{}}

\usepackage{fullpage}
\usepackage[dvipsnames]{xcolor}

\usepackage{times}

\usepackage{dsfont}
\usepackage{amsthm}

\usepackage{thmtools}
\usepackage{thm-restate}

\usepackage{graphicx}
\usepackage{url}
\usepackage{nicefrac}
\usepackage{enumitem}
\usepackage{authblk}
\usepackage{tikz}
\usepackage{subcaption}
\usepackage{xspace}
\usepackage[compact]{titlesec}
\usetikzlibrary{arrows,decorations.markings}
\usepackage[normalem]{ulem}
\usepackage{comment}

\usepackage{algorithm}
\usepackage[noend]{algpseudocode}

\usepackage{typearea}

\usepackage{pgf}
\usetikzlibrary{decorations.pathmorphing}
\usetikzlibrary{patterns}
\usetikzlibrary{backgrounds}
\usetikzlibrary{scopes}
\usetikzlibrary{arrows, automata, positioning}
\usetikzlibrary{decorations.markings}
\usepackage{calc}
\usepackage{color}
\usepackage{lipsum}

\usepackage{color}

\definecolor{blueblack}{rgb}{0,0,.7}

\newcounter{sideremark}

\definecolor{Darkblue}{rgb}{0,0,0.4}
\definecolor{Brown}{cmyk}{0,0.61,1.,0.60}
\definecolor{Purple}{cmyk}{0.45,0.86,0,0}
\definecolor{brickred}{rgb}{0.8, 0.25, 0.33}
\usepackage{cleveref}

\theoremstyle{plain}
\theoremstyle{definition}
\theoremstyle{remark}
\newtheorem*{remark*}{Remark}

\newcommand{\RR}{{\mathbb R}}

\newcommand{\NN}{{\mathbb N}}
\newcommand{\EE}{{\mathbb E}}

\newcommand{\eps}{\varepsilon}

\newcommand{\calB}{\mathcal{B}}
\newcommand{\calC}{\mathcal{C}}
\newcommand{\calE}{\mathcal{E}}

\newcommand{\calS}{\mathcal{S}}

\newcommand{\calF}{\mathcal{F}}

\newcommand{\calI}{\mathcal{I}}

\newcommand{\OPT}{\mathsf{OPT}}

\global\long\def\calS{\mathcal{S}}%

\newcommand{\APX}{APX}

\global\long\def\Jt{J_{\mathrm{long}}}%
\global\long\def\Jshort{J_{\mathrm{short}}}%
\global\long\def\Jbnd{J_{\mathrm{bnd}}}%

\newcommand{\greedy}{\textsc{Greedy}}

\global\long\def\E{\mathbb{E}}%

\DeclareMathOperator{\tw}{tw}

\usepackage[colorinlistoftodos,textsize=tiny,textwidth=2cm,color=green!50!gray]{todonotes}

\ifdefined\DEBUG
  \newcommand{\aw}[1]{\textcolor{ForestGreen}{#1}}
  \newcommand{\ant}[1]{\textcolor{blue}{#1}}
  \newcommand{\fab}[1]{\textcolor{red}{#1}}
  \newcommand{\ale}[1]{\textcolor{violet}{#1}}

\newcommand{\alnote}[1]{\todo[color=violet!100!black!50]{AA: #1}}
\newcommand{\anote}[1]{\todo[color=blue!100!black!50]{AT: #1}}
\newcommand{\awr}[1]{\todo[color=green!100!black!50]{AW: #1}}
\newcommand{\fabr}[1]{\todo[color=red!100!black!50]{F: #1}}

\else
  \newcommand{\ant}[1]{#1}
  \newcommand{\fab}[1]{#1}
  \newcommand{\aw}[1]{#1}
  \newcommand{\ale}[1]{#1}

  \newcommand{\anote}[1]{}
  \newcommand{\fabr}[1]{}
  \newcommand{\awr}[1]{}
  \newcommand{\alnote}[1]{}
\fi

\title{Improved Approximation Algorithms for\\ Non-Preemptive Throughput Maximization}

\author[1]{Alexander Armbruster\thanks{Funded by the Deutsche Forschungsgemeinschaft (DFG, German Research Foundation) – project 551896423.}}
\author[2]{Fabrizio Grandoni\thanks{Partially supported by the Swiss National Science Foundation (SNSF) Grants 200021-200731/1 and 200021-236706.}}
\author[2]{Antoine Tinguely\thanks{Supported by the Swiss National Science Foundation (SNSF) Grant 200021-200731/1.}}
\author[1]{Andreas Wiese}

\affil[1]{Technical University of Munich, Munich, Germany

\texttt{\small alexander.armbruster@tum.de}, \texttt{\small andreas.wiese@tum.de}}
\affil[2]{USI-SUPSI, IDSIA, Lugano, Switzerland

\texttt{\small fabrizio@idsia.ch}, \texttt{\small antoine.tinguely@tum.de}}

\makeatother

\providecommand{\definitionname}{Definition}
\providecommand{\lemmaname}{Lemma}
\providecommand{\propositionname}{Proposition}
\providecommand{\theoremname}{Theorem}
\providecommand{\corollaryname}{Corollary}

\begin{document}
\global\long\def\NN{\mathbb{N}}%
\global\long\def\tw{\mathrm{tw}}%
\global\long\def\eps{\varepsilon}%
\global\long\def\OPT{\mathrm{OPT}}%
\global\long\def\calB{\mathcal{B}}%
\global\long\def\calS{\mathcal{S}}%
\global\long\def\Kavg{K_{\mathrm{avg}}}%
\global\long\def\calC{\mathcal{C}}%
\global\long\def\Js{J_{\mathrm{local}}}%
\global\long\def\Jl{J_{\mathrm{global}}}%

\maketitle



\begin{abstract}
\noindent The (Non-Preemptive) Throughput Maximization problem is a natural and fundamental scheduling problem. We are given $n$ jobs, where each job $j$ is characterized by a processing time and a time window, contained in a global interval $[0,T)$, during which~$j$ can be scheduled. Our goal is to schedule the maximum possible number of jobs non-preemptively on a single machine, so that no two scheduled jobs are processed at the same time. This problem is known to be strongly NP-hard. The best-known approximation algorithm for it has an approximation ratio of 
$1/0.6448 + \eps \approx 1.551 + \eps$ [Im, Li, Moseley IPCO'17], improving on an earlier result in 
[Chuzhoy, Ostrovsky, Rabani FOCS'01]. In this paper we substantially improve the approximation factor for the problem to $4/3+\eps$ for any constant~$\eps>0$. Using pseudo-polynomial time $(nT)^{O(1)}$, we improve the factor even further to $5/4+\eps$. \fab{Our results extend to \aw{the setting in which we are given} an arbitrary number of (identical) machines.}
\end{abstract}

\thispagestyle{empty}
\newpage

\setcounter{page}{1}

\section{Introduction}

In this paper we study the \emph{(Non-Preemptive) Throughput Maximization} problem
(also known as \emph{Job Interval Scheduling}) which is one of the
most basic scheduling problems. We are given a set of $n$ jobs~$J$,
where each job $j\in J$ is characterized by its \emph{processing
time} $p_{j}\in\NN$, its \emph{release time} $r_{j}\in\NN$, and
its \emph{deadline} $d_{j}\in\NN$. For each job $j\in J$ we define its
\emph{time window }by $\tw(j)\coloneqq[r_{j},d_{j})$. The goal is
to select a subset $J'\subseteq J$ of the jobs and to compute a non-preemptive
schedule for $J'$ on one machine.
More formally, we seek to compute
a start time 
\awr{line had disappeared}$s(j)\in\NN$ for each job $j\in J'$ such that $[s(j),s(j)+p_{j})\subseteq\tw(j)$, meaning that
we execute $j$ during $[s(j),s(j)+p_{j})$.
We require that for any two distinct jobs $j,j'\in J'$ their intervals $[s(j),s(j)+p_{j})$
and $[s(j'),s(j')+p_{j'})$ are disjoint. The objective is to maximize
the number of scheduled jobs, i.e., to maximize $|J'|$. Not surprisingly,
the problem \ale{and} its variants and generalization have several applications,
see, e.g., \cite{hall1994maximizing,blazewicz2013scheduling,fischetti1987fixed,lawler1993sequencing}
and references therein.

Throughput Maximization is (strongly) NP-hard \cite{pinedo08,garey1977two} which motivates studying approximation
algorithms for it. However, from this point of view, it is not very
well understood. The currently best known approximation factor in
polynomial time (and even in pseudo-polynomial or quasi-polynomial
time) is $1/0.6448+\eps \approx 1.551+\eps$ for any constant $\eps>0$, due to Im, Li and Moseley
\cite{ILM17,im2020breaking}. This
improves a previous result by Chuzhoy,
Ostrovisky and Rabani \cite{COR01,chuzhoy2006approximation} which
achieves an approximation ratio of $\frac{e}{e-1}+\eps \approx 1.582+\eps$
(in fact, even for a more general version
where for each job we are given an explicit set of possible execution intervals, instead of our implicitly
defined intervals of the form $[s(j),s(j)+p_{j})\subseteq\tw(j)$) and other previous results \cite{spieksma1999approximability,bar2001unified,bar2001approximating,berman2000multi}.

However, Throughput Maximization is \emph{not} known to be APX-hard
and, hence, it might still admit a PTAS! We find intriguing that the
approximability status of such a basic problem is still rather unclear.

\subsection{Our Results and Techniques}

In this paper we substantially improve the best known approximation
ratio for Throughput Maximization. More specifically, we obtain the
following two main results.
\begin{thm}
\label{thr:mainPoly} For any constant $\eps>0$, there is a polynomial-time
randomized $(4/3+\eps)$-approximation algorithm for Throughput Maximization.
\end{thm}

Using pseudo-polynomial time, we can do even better. We define $T:=\max_{j\in J}d_{j}$
and observe that, hence, each job $j\in J$ must be scheduled within
the interval $[0,T)$.
\begin{thm}\label{thr:mainPseudopoly}
For any constant $\eps>0$, there is a randomized $(5/4+\eps)$-approximation
algorithm for Throughput Maximization with a running time of $(nT)^{O_{\eps}(1)}$.
\end{thm}

In the following, we illustrate the main ideas behind our results. Our
starting point is the approach by Chuzhoy et al. \cite{chuzhoy2006approximation}.\awr{strictly speaking, we simplify the thing witih the full blocks here, but I think this is ok}
The authors present a polynomial-time procedure that partitions the
time horizon $[0,T)$ into a collection of at most $\eps|\OPT|$ intervals
which we will call \emph{blocks}. 
The blocks are defined such that
there is a $(1+\eps)$-approximate solution in which
\begin{itemize}
\item each job is scheduled entirely within one block and
\item inside each block at most $O_{\eps}(1)$ jobs are scheduled.
\end{itemize}
Based on the blocks, they define a configuration-LP that has a configuration
for each combination of a block $B$ and a set of at most $O_{\eps}(1)$
jobs that can be scheduled within $B$; in particular, the LP has a polynomial
number of variables and constraints.

Given an optimal fractional solution to the configuration-LP, they
sample independently one configuration for each block. It might happen
that some job $j\in J$ appears in more than one sampled configuration;
in this case, it can still be scheduled only once in the computed solution.
However, one can show that if a job $j\in J$ appeared fractionally
in $y_{j}^{*}$ configurations of the optimal LP-solution, then it
appears in at least one sampled configurations with probability at
least $\frac{e-1}{e}y_{j}^{*}$ which yields the mentioned approximation
ratio of $\frac{e}{e-1}+\eps\approx1.58\fab{2}+\eps$.

In both of our algorithms, we use a similar configuration-LP, but
we define the blocks in a different way and also invoke a different
rounding procedure. We start with our $(4/3+\eps)$-approximation
algorithm. Like \fab{in} the algorithm by Chuzhoy et al. \cite{chuzhoy2006approximation},
we sample a configuration for each block according to the optimal
LP-solution. However, we do not directly assign the jobs according
to the sampled configurations. Instead, if a sampled configuration
schedules a job $j$ during some time interval $[s(j),s(j)+p_{j})$,
then we interpret this interval 
as a \emph{slot}
during which we might schedule \emph{some} job $j'$ whose time window
and processing time allow to process it completely during $[s(j),s(j)+p_{j})$ (possibly $p_{j'}<p_j$ and then the machine would remain idle during some parts of the interval).
Then, we use a bipartite matching routine to compute the largest set
of jobs that can be assigned in this way to the slots, i.e., each
slot gets at most one compatible job assigned to it and each job is assigned to at most one slot. In particular,
if a job $j$ appears in two sampled configurations, then it creates
two slots such that we can assign $j$ to one of them and potentially
another job $j'$ to the other slot. In contrast, in \cite{chuzhoy2006approximation}
the second slot was kept empty in this case and, hence, it was lost. On a high level, our algorithm might
not seem too different from \cite{chuzhoy2006approximation}; however,
we show that in expectation our resulting matching yields an approximation
ratio of only $4/3+\eps < 1.334+\eps$.

One key idea for proving this is a more sophisticated definition of
our blocks. We ensure that they still have the properties described
above. In addition, we define a second partition of $[0,T)$ into
\emph{superblocks}, \fab{where} 
each superblock \fab{is the union of} a large (constant) number of consecutive blocks. We show
that we can compute such a partition for which there is a $(1+\eps)$-approximate
solution with a set of jobs $\OPT'$ in which each job $j\in\OPT'$
is scheduled either (see Figure~\ref{fig:block-superblock}):
\begin{itemize}
\item in the leftmost or the rightmost block that $\tw(j)$ intersects;
we call these blocks the \emph{boundary blocks} \fab{for} $j$,\fabr{Around we use for and of, but I think for is more frequent} or
\item in one of the superblocks $S$ \emph{spanned} by $j$, i.e., such that $S\subseteq \tw(j)$.
\end{itemize}

\begin{figure}
    \centering
    \footnotesize
    \includegraphics[width=\textwidth]{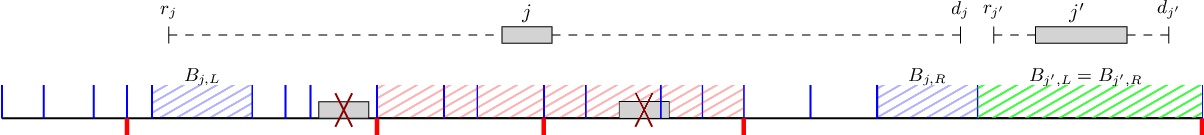}
    \caption{The blue and red lines delimitate the blocks and superblocks, respectively. Job $j$ is global and can be scheduled within its boundary blocks $B_{j,L}$ and $B_{j,R}$ (marked blue) or within one of the two superblock\fab{s} that it spans (marked red). It cannot be scheduled between a spanned superblock and $B_{j,L}$ or $B_{j,R}$, and cannot intersect two (or more) blocks. Job $j'$ is local. \fab{Its (unique) boundary block is marked in green.}}
    \label{fig:block-superblock}
\end{figure}

\noindent In order to prove that there exists a sufficiently large matching, a particularly
simple case is when in the optimal solution to the configuration-LP,
each job $j\in J$ appears only in configurations corresponding to
its boundary blocks. Then, we can easily show that it appears in
at least one sampled block-configuration with probability at least
$\frac{3}{4}y_{j}^{*}$. Hence, we can we simply match each job $j$
to one of ``its own'' slots like \cite{chuzhoy2006approximation}
and obtain an approximation ratio of $4/3+\eps$.

However, we need to show that we can achieve the same approximation factor also when the jobs might be fractionally scheduled in non-boundary blocks. This is
substantially more complex, and one of the main contributions of this paper.
At a high level, we show that in expectation there exists a large-enough fractional (bipartite) matching between jobs and slots. Standard matching theory then implies that the maximum matching is also large enough. 

In more detail, if a job $j$ belongs to the sampled configuration for at least one of its boundary blocks, say $B$, we simply integrally match $j$ with the slot created by $j$ in $B$. The remaining jobs are fractionally matched to slots in \aw{the superblocks spanned by them}.
We need to show that the latter
jobs yield sufficiently many fractional edges. To that aim, we use a fractional version of the classical harmonic grouping technique and concentration arguments that critically exploit the fact that each of our superblocks contains many blocks. For each superblock $S$, we apply harmonic grouping
based on the fractional solution to the configuration-LP for its contained
blocks. This yields $O(1/\eps)$ groups of jobs where each group contains
essentially the same (fractional) number of jobs, and the processing times of the jobs in one group are not larger than the processing times of the next group. Notice that in harmonic
grouping one typically forms groups of items of similar cardinality,
while we form the groups so that they have a similar \emph{fractional
cardinality} in terms of the computed LP-solution.  Since the configuration
for each of the (many) blocks contained in $S$ were sampled independently,
we can show via concentration arguments that, with sufficiently large probability, for each of the $O(1/\eps)$
job groups, the number of sampled slots is essentially the same as
the fractional number of jobs in that group \fab{(\aw{and, hence, also in the next}\awr{it said ``previous group''} group)}. 
\fab{When this event happens}, we \fab{fractionally} match the jobs from each job
group to the slots of the next job group \fab{(with larger processing time)}. For this it is crucial
that for each matched job $j$ its time window $\tw(j)$ contains
the entire superblock $S$; due to this property, we can schedule $j$ arbitrarily
within $S$. 

There is a subtle technical issue. In order to achieve the claimed approximation factor, we use that the fractional amount by which we assign a job $j$ to slots in a superblock $S$ spanned by $j$ depends on the sampled configurations for the boundary blocks for~$j$. More specifically, if $j$ belongs to one such configuration, it cannot be also fractionally matched to a slot in a spanned superblock. Hence, these fractional amounts for the different jobs are not independent random variables, \aw{since two jobs might have the same boundary block}! Therefore, we cannot use the standard Chernoff's bound to prove that the relevant random variables are sufficiently concentrated around their respective means. However, we are able to show that the impact of such dependencies is sufficiently small. More specifically, we can still use the concentration bounds for \emph{read}-$k$ families of random variables in \cite{gavinsky2015tail} to obtain the desired properties.

A closer look into our analysis reveals that our approximation ratio
is even $1+\eps$ (compared to the optimal LP solution) with respect to the
profit of jobs $j$ whose time
window $\tw(j)$ is contained in some block; we call such jobs \emph{local}.
For the other (\emph{global}) jobs, our approximation ratio is $4/3+\eps$, again
compared to the optimal fractional solution. Therefore, we design
a second rounding routine that achieves an approximation ratio of
$1+\eps$ w.r.t.~the global jobs (but does not schedule any local
job). The best of the two solutions then has an approximation ratio
of only~$5/4+\eps$.

For this second algorithm, we need a different
block decomposition in which each block contains intuitively $\Theta_{\eps}(\log T)$
jobs on average from some near-optimal solution. However, then our
configuration-LP has $n^{\Theta_{\eps}(\log T)}$ variables and
we can no longer solve it trivially in polynomial time. Therefore,
we use the ellipsoid method with a new separation oracle for the dual
LP. The separation problem is a \emph{weighted} generalization of our (initial) Throughput Maximization problem, however with the additional constraint that any solution can contain \emph{at most} 
$\Theta_{\eps}(\log T)$ jobs.
We show how to solve the latter problem in pseudopolynomial time with a suitable dynamic program and the 
color coding technique \cite{AYZ95}.

Another major difference in our second algorithm is the LP-rounding procedure. Instead of independently sampling a configuration for each block, we independently assign a job to a block $B$ according
to the marginal probabilities induced by the fractional solution.
It might
be that there is a block $B$ whose assigned jobs cannot all be scheduled
within $B$. However, we show that if we remove $O_{\eps}(\log T)$ intuitively
relatively ``long'' jobs from each block $B$, then due to concentration
arguments the remaining jobs can (most likely) be scheduled within
$B$. More precisely, we remove a job $j$ if its processing time
$p_{j}$ is relatively large compared to the length of its time window
within $B$, i.e., $p_{j}$ is large compared to the length of $B\cap\tw(j)$.
Since our blocks contain $\Theta_\eps(\log T)$ jobs on average from
some near-optimal solution, we can argue that the removed jobs determine
only an $\eps$-fraction of the optimal profit.

We remark that a similar random assignment was used by Im, Li and
Moseley \cite{im2020breaking}. However, they had to disallow
to schedule a job in its boundary blocks which lost the corresponding
profit of the LP. In contrast, since we allow a block to contain up to ${O_{\eps}}(\log T)$ jobs and we remove the relatively long jobs, 
we \emph{can} assign a job to its boundary intervals too
and, hence, we obtain a much better approximation ratio of $5/4+\eps$.

Our results extend to the generalization of Throughput Maximization on $m$ (identical) machines (see Section~\ref{sec:multiple} for details).
\begin{thm}
\label{thr:multipleMachines} For any constant $\eps>0$, there is a polynomial-time
randomized $(4/3+\eps)$-approximation algorithm and a pseudopolynomial-time randomized $(5/4+\eps)$ approximation algorithm for Throughput Maximization on $m$ machines.
\end{thm}

\subsection{Other related Work}
\label{sec:related}

The first approximation algorithm for Throughput Maximization was a $2$-approximation algorithm due to 
Spieksma \cite{spieksma1999approximability}.
While the problem is NP-hard in general \cite{pinedo08,garey1977two}, it can be solved in polynomial time if all the processing times are identical \cite{GareyJST81} and in pseudo-polynomial time if preemption is allowed \cite{lawler1990dynamic}. The latter setting also admits an FPTAS \cite{pruhs2007approximation}.

\awr{\cite{bar2001approximating} has a $2+\eps$ approximation for the weighted time windows case and a factor $2$ in the discrete case; I think their key contribution were the same results for unrelated machines}
There are natural generalizations of Throughput Maximization to multiple machines and/or to the weighted case in which each job has a weight and we wish to maximize the total weight of scheduled jobs.
There is an algorithm by Bar-Noy, Guha, Naor and Schieber \cite{bar2001approximating} that achieves a $\frac{(1+1/m)^m}{(1+1/m)^m-1}$-approximation for any number of machines $m$ (also in the weighted case in pseudo-polynomial time)\awr{they also have a worse approximation ratio in polytime}; note that this ratio converges to $\frac{e}{e-1}$ for $m \rightarrow \infty$. 
Essentially the same ratio was obtained in \cite{berman2000multi} in polynomial time. There is also a $(2+\eps)$-approximation algorithm for the weighted case on multiple machines using the local ratio technique \cite{bar2001unified}.
\awr{in \cite{berman2000multi} there is a polynomial time algorithm with a ratio of essentially $\frac{(1+1/m)^m}{(1+1/m)^m-1-\eps}$}
 The algorithm by Chuzhoy et al. \cite{chuzhoy2006approximation} yields the ratio of $\frac{e}{e-1}+\varepsilon$ for any number of machines but only in the unweighted case. The algorithm by Im et al.~\cite{im2020breaking} yields the mentioned approximation ratio of $1/0.6448+\eps<1.551+\eps$ for any number of machines in the unweighted case, and the authors provide also an $(1+\eps)$-approximation for the weighted case if the number of machines $m$ is sufficiently large (for any constant $\varepsilon >0$). If we allow resource augmentation, i.e., shortening the processing time of each input job by a factor of $1+\eps$, there is even a $(1+\varepsilon)$-approximation in quasi-polynomial time known in the weighted case and for any number of machines \cite{dynamic-programming-framework}.

The Unsplittable Flow on a Path with Time Windows problem (UFPTW) is another generalization of Throughput Maximization where each job is in addition specified by a certain \emph{demand} for a shared resource whose capacity might vary over time. In this problem, multiple jobs can be processed at the same time, provided that their total demand is within the available capacity at that time. Throughput Maximization with $m$ machines is the special case where all the demands are $1$, and the available capacity is uniformly $m$. For UFPTW there is a polynomial time $O(\log n / \log \log n)$-approximation algorithm for the weighted case and an $O(1)$-approximation in the unweighted case \cite{grandoni2015improved}, improving on \cite{CCGRS14}. These results hold for a more general setting of explicitly given possible execution intervals and corresponding demands, similar to~\cite{chuzhoy2006approximation}. Also, there is a quasi-polynomial time $(2+\varepsilon)$-approximation algorithm under resource augmentation for the weighted case \cite{armbruster2025approximability}. 
In \cite{armbruster2025approximability} it is also shown that the
problem is APX-hard (even in the unweighted setting); the hardness reduction critically exploits \aw{that the tasks can have different} demands, and, hence, it does not extend to Throughput Maximization. The special case of the above problem where the length of each time window equals the respective processing time is the classical and very well-studied Unsplittable Flow on a Path problem \cite{BCES2006,BFKS09,AGLW14,BSW14,BGKMW15,GMWZ17,GMWZ18,GMW21,GMW22} which admits a PTAS \cite{GMW22STOC}.

Other scheduling problems on one or multiple machines include for example makespan minimization \cite{hochbaum1987using,jansen2010eptas,jansen2020closing},
minimizing the average (weighted) completion times of the given jobs \cite{chekuri2001approximation}, minimizing the average (weighted) job's flowtimes \cite{armbruster2023ptas}, or scheduling jobs with even more general cost functions \cite{DBLP:journals/siamcomp/BansalP14, cheung2017primal,antoniadis2017qptas,hohn2018unsplittable}.
We refer to \cite{agnetis2025fifty,brucker2004scheduling,leung2004handbook,lawler1993sequencing} for overviews on the scheduling literature.

\section{Polynomial time approximation algorithm}\label{sec:polyAlgo}

In this section, we present our polynomial time $(4/3+\eps)$-approximation
algorithm. We assume w.l.o.g.~that $\min_{j\in J}r_{j}=0$ and recall that
$T=\max_{j\in J}d_{j}$. Let $\eps>0$ be a sufficiently small \fab{constant} and assume that $1/\eps\in \NN$. For each $k\in\NN$, we define $[k]:=\{1,\dots,k\}$.

In Section \ref{sec:blockConstruction}, we present a method to partition $[0,T)$ into subintervals
such that there exist \fab{a} $(1+\eps)$-approximate solution that \fab{is}
``aligned'' with this partition and satisfies certain additional properties
that we will exploit later in our approximation algorithm. Based on this, in Section~\ref{sec:configurationLP} we define a modified configuration LP. In Section \ref{sec:firstRounding} we present our rounding algorithm, which is then analyzed in Section \ref{sec:polynomialTimeAnalysis}.

\subsection{Construction of blocks and superblocks  }\label{sec:blockConstruction}

As\fabr{Removed repetition about $\eps$} a first step, we invoke a method from \cite{chuzhoy2006approximation}
to partition $[0,T)$ into \emph{blocks} where we define a block to
be an interval of the form $[a,b)$ for some $a,b\in[0,T\fab{]}$\fabr{Was )} with $a<b$.
Given \fab{a set} of jobs $J'\subseteq J$ and \fab{a} corresponding
schedule with a starting time $s(j)\in \{0,\ldots,T\}$ for each job $j\in J'$, we say that a job $j\in J'$ is \emph{scheduled in a block $B$} if $[s(j),s(j)+p_{j})\subseteq B$. We remark that the running time bound in \cite{chuzhoy2006approximation} holds for a discrete version of the problem and that
property \ref{item:cor-bound-opt} below is not explicitly stated in \cite{chuzhoy2006approximation}. For the sake of completeness, we sketch a proof of the following lemma in the appendix.
\begin{lem}
\label{lem:block-preprocessing}
Assume that $|\OPT|\geq (\frac{6}{\eps})^3$. In polynomial time 
we can compute a partition of $[0,T)$ into a set of blocks $\calB_{0}$ such that there exists a \fab{set of jobs} $\OPT''\subseteq J$ \fab{and a feasible schedule of them} with the following properties:
\begin{enumerate}[label=(A\arabic*)]
\item\label{item:cor-opt} $|\OPT''|\geq(1-\eps)|\OPT|$, 
\item\label{item:cor-restricted} each job $j\in\OPT''$ is scheduled in some block $B\in\calB_{0}$,
\item\label{item:cor-block-bound} for each block $B\in\calB_{0}$ there are at most $K_{0}=(1/\eps)^{O(1/\eps\log(1/\eps))}$
jobs of $\OPT''$ that are scheduled in $B$, 
\item\label{item:cor-bound-opt} $|\OPT''|\ge|\calB_{0}|/\eps$.

\end{enumerate}
\end{lem}
To avoid ambiguity later, we refer to the blocks from Lemma~\ref{lem:block-preprocessing} as \emph{elementary blocks.}
We remark that property \ref{item:cor-bound-opt} of Lemma~\ref{lem:block-preprocessing}
implies that the \emph{average} number of jobs in an elementary block
is at least $1/\eps$, while by property \ref{item:cor-block-bound} the \emph{maximum
}number of jobs in each elementary block is bounded by $K_{0}=O_{\eps}(1)$. 

As a next step, we want to compute a partition of $[0,T)$ with stronger
properties. Instead of just blocks, we want to compute superblocks
such that each superblock is the union of a set of consecutive blocks (\aw{which will be} constantly
many blocks below).
\begin{defn}
A pair $(\calB,\calS)$ is a \emph{block-superblock partition} if
both $\calB$ and $\calS$ are sets of blocks that form a partition
of $[0,T)$ and each block $S\in\calS$ is the union of a set of consecutive blocks
$\calB(S) \subseteq \calB$. We call the blocks in $\calS$ also \emph{superblocks.}
\end{defn}

Intuitively, we use the blocks in $\calB_{0}$ as a basis and, with
a suitable shifting argument, glue them together to form the blocks
$\calB$. We use a parameter $\Delta$ that controls how many elementary blocks
we glue together here. For our polynomial-time algorithm in this section we will define $\Delta:=1$, while for the pseudopolynomial-time algorithm in the next section we will use $\Delta=O_{\eps}(\log T)$.
Then, we define each superblock
as the union of (constantly) many consecutive blocks in $\calB$. Thanks to \fab{the mentioned}
shifting argument, we can ensure that there is a $(1+\eps)$-approximate
solution in which each job $j$ \aw{may be scheduled only} within a specific 
set of blocks and superblocks. \fab{In more detail,} given a block-superblock partition
$(\calB,\calS)$, for each job $j\in J$ we define its \emph{release
block }$B_{j,L}\in\calB$ such that $r_{j}\in B_{j,L}$ and its \emph{deadline
block} $B_{j,R}=[s,t)\in\calB$ such that $d_{j}\in (s,t]$. The release block and the deadline block of $j$ are also called its \emph{boundary} blocks. We remark that possibly $B_{j,L}=B_{j,R}=B$, in which case $\tw(j)\subseteq B$. Furthermore,
we say that $j$ \emph{spans} a block $B\in\calB$ or a superblock
$S\in\calS$ (and $B$ or $S$ are \emph{spanned} by $j$) if $B\subseteq\tw(j)$ or $S\subseteq\tw(j)$, resp. In the following lemma we use the value $K_0$ as defined in Lemma \ref{lem:block-preprocessing}.

\begin{lem}
\label{lem:good-block-superblock-partition}Let $\Delta\in\NN$ be a given parameter, and assume that $|\OPT|\geq 5\ant{K_0}\Delta(2K_0/\eps^5)^{1/\eps}$.
\anote{we need factor $K_0$ in the bound}
In time $\Delta^{O(1)} n^{O_\eps(1)}$ we can compute at most $1/\eps$ block-superblock partitions and for each one of them a value $K$ with \ant{$K\le \Delta(2K_0/\eps^5)^{1/\eps} = O_{\varepsilon}(\Delta)$}
such that for at least one computed partition $(\calB,\calS)$ there exists \fab{a set of jobs}
$\OPT'\fab{\subseteq J}$ and a feasible schedule of them with the following properties:
\begin{enumerate}[label=(B\arabic*)]
\item\label{item:block-superblock-partition-lb} $|\OPT'|\geq (1-2\varepsilon)|\OPT|$,
\item\label{item:border-or-spanning} each job $j\in\OPT'$ is scheduled in $B_{j,L}$, or in $B_{j,R}$,
or in a block $B\in\calB$ for which $j$ spans the superblock $S\in\calS$
containing $B$,
\item\label{item:configuration-size} for each block $B\in\calB$ there are at most $K$ jobs from $\OPT'$ that are scheduled in $B$, 
\item\label{item:fewBlocks} $|\OPT'|\ge \frac{\Delta}{\eps}|\calB|$ and $|\OPT'|\ge \frac{K}{\eps^6}|\calS|$.
\end{enumerate}
\end{lem}

\begin{proof}
Consider the partition into elementary blocks $\calB_0$ provided by Lemma~\ref{lem:block-preprocessing}.
\ant{From the assumption $|\OPT| \geq 5\ale{K_0}\Delta (2K_0/\eps^5)^{1/\eps}$ and using \ref{item:cor-opt} and $\varepsilon \leq 1/5$, we have
\begin{equation*}
        |\calB_0| \geq \frac{1}{K_0} \cdot |\OPT''|
        \geq \frac{(1-\varepsilon)}{K_0} \cdot |\OPT|
        \geq 5\Delta (1-\varepsilon) (2K_0/\eps^5)^{1/\eps}
        \geq 4\Delta (2K_0/\eps^5)^{1/\eps}.
\end{equation*}
}
We add up to $2\Delta(2K_0/\eps^5)^{1/\eps}=O_\eps(\Delta)$ dummy blocks of type $[T,T)$ so as to obtain a set $\calB'_0$ of elementary blocks whose cardinality is a multiple of $2\Delta(2K_0/\eps^5)^{1/\eps}$. Notice that $|\calB'_0|\leq \frac{3}{2}|\calB_0|$.
We define the partition into blocks $\calB_1$ by merging $2\Delta$ consecutive blocks in $\calB'_0$, and define recursively $\calB_\ell$ for $\ell\in \{2,\ldots,1/\eps\}$ by merging $2K_0/\eps^5$ consecutive blocks of $\calB_{\ell-1}$. For $\ell\in [1/\eps]$, we define the partition $\calS_\ell$ into superblocks by merging $2K_0/\eps^5$ consecutive blocks of $\calB_{\ell}$. The pairs $(\calB_\ell,\calS_\ell)$ for $\ell\in [1/\eps]$ are the block-superblock partitions in the claim, and we set the corresponding value of $K$ to $K_\ell \coloneqq 2\Delta K_0(2K_0/\eps^5)^{\ell-1} \leq \Delta(2K_0/\eps^5)^{1/\eps}$, as required. We also observe that
\begin{itemize}
    \item $|\calB_\ell|=\frac{1}{2\Delta(2K_0/\eps^5)^{\ell-1}}|\calB'_0|$;
    \item $|\calS_\ell|=\frac{1}{2K_0/\eps^5}|\calB_{\ell}|$;
    \item $\calS_\ell=\calB_{\ell+1}$ for $\ell\in [1/\eps-1]$.
\end{itemize}
Consider the solution $\OPT''$ as in Lemma~\ref{lem:block-preprocessing} with the associated schedule. We first observe that, for each $j\in \OPT''$, there exists at most one partition $(\calB_{\ell},\calS_{\ell})$ such that property \ref{item:border-or-spanning} is not satisfied. Indeed, consider one such level $\ell$. Let $B^\ell_{j,L}$ and $B^\ell_{j,R}$ (resp, $S^\ell_{j,L}$ and $S^\ell_{j,R}$) be the corresponding boundary blocks (resp, superblocks). The violation of the property implies that $j$ is either scheduled inside $S^\ell_{j,L}$ but not in $B^\ell_{j,L}$, or inside $S^\ell_{j,R}$ but not in $B^\ell_{j,R}$. Assume w.l.o.g. that the first case applies. Then $j$ is scheduled within $B^{\ell'}_{j,L}$ for every $\ell'>\ell$. At the same time, $j$ is scheduled inside a spanning superblock of $(\calB_{\ell'},\calS_{\ell'})$ for every $\ell'<\ell$. In both cases property \ref{item:border-or-spanning} is satisfied.
Let $D_{\ell}$ be the jobs in $\OPT''$ that do not satisfy property \ref{item:border-or-spanning} w.r.t. $(\calB_{\ell},\calS_{\ell})$. By the above discussion such sets are disjoint, and therefore one such set $D_{\ell^*}$ contains at most $\eps |\OPT''|$ jobs. We define $\OPT':=\OPT''\setminus D_{\ell^*}$, and for these jobs preserve the same schedule as for $\OPT''$.

We claim that the partition $(\calB_{\ell^*},\calS_{\ell^*})$ together with $\OPT'$ and the associated schedule satisfy all the desired properties. Trivially, $|\OPT'|\geq (1-\eps)|\OPT''|\geq (1-\eps)^2|\OPT|\geq (1-2\eps)|\OPT|$, hence property \ref{item:block-superblock-partition-lb} holds. Each $j\in \OPT'$ is scheduled inside an elementary block, hence inside a block in $\calB_{\ell^*}$. By construction $j$ satisfies property \ref{item:border-or-spanning}, otherwise it would belong to $D_{\ell^*}$. Each block of $\calB_{\ell^*}$ consists of $2\Delta(2K_0/\eps^5)^{\ell^*-1}$ elementary blocks, and at most $K_0$ jobs of $\OPT''$ (hence of $\OPT'$) are scheduled in each elementary block. Thus at most $K_{\ell^*}=2\Delta K_0 (2K_0/\eps^5)^{\ell^*-1}$ jobs of $\OPT'$ are scheduled inside each block of $\calB_{\ell^*}$ as required by property \ref{item:configuration-size}.  It remains to prove property \ref{item:fewBlocks}. \ant{With $\varepsilon \leq 1/4$,}
$$
|\OPT'|\geq (1-\eps)|\OPT''|\geq \frac{2(1-\eps)}{3\eps}|\calB'_0|\geq \frac{\Delta}{\eps} (2K_0/\eps^5)^{\ell^*-1}|\calB_{\ell^*}| \geq \frac{\Delta}{\eps}|\calB_{\ell^*}|. 
$$
Furthermore,  
$
\frac{2K_0K_{\ell^*}}{\eps^5}|\calS_{\ell^*}|=K_{\ell^*} \cdot|\calB_{\ell^*}|=K_0 \cdot |\calB'_0|,
$ 
hence
$$
|\OPT'|\geq \frac{2(1-\eps)}{3\eps}|\calB'_0|\geq \frac{4(1-\eps)}{3\eps} \cdot\frac{K_{\ell^*}}{\eps^5}|\calS_{\ell^*}|\geq \frac{K_{\ell^*}}{\eps^6}|\calS_{\ell^*}|.
$$

\end{proof}

In the rest of this section, we will assume $\Delta=1$ \ant{and w.l.o.g.\ $|\OPT| \geq 5K_0\Delta(2K_0/\eps^5)^{1/\eps}$, since otherwise
\anote{need to add $K_0$ anywhere else?}
 the problem  can be solved exactly \aw{in polynomial time} by enumeration.} Therefore, we
can apply Lemma \ref{lem:good-block-superblock-partition}. Intuitively, among the $1/\eps$ computed block-superblock partitions, we guess one for which \fab{p}roperties \ref{item:block-superblock-partition-lb}-\ref{item:fewBlocks} hold. Formally, we execute the following steps for each of them and finally output the best obtained solution.
Hence, in the analysis we may assume that we know the 
partition $(\calB,\calS)$, together with its corresponding value $K$, for which the described solution $\OPT'$
exists. 

\subsection{The linear program}
\label{sec:configurationLP}

We define a configuration-LP which intuitively computes a (fractional)
solution that satisfies~\ref{item:border-or-spanning} and \ref{item:configuration-size}.
A configuration $C$ is specified by a pair $(B_C,J_C)$, where $B_C=[s,t)$ is a block and $J_C$ is a subset of jobs that can be feasibly scheduled inside $B_C$ according to Lemma~\ref{lem:good-block-superblock-partition}; formally, we require that $|J_C|\leq K$ and that, for each $j\in J_C$, the block $B_C$ is either a boundary block for $j$ or $B_C$ is contained in a superblock $S$ spanned by $j$ (i.e., $S\subseteq \tw(j)$). We denote by $s_C:J_C\rightarrow \{s,\ldots,t-1\}$ an arbitrary but fixed feasible schedule of $J_C$ inside $B_C$.
For each block $B$ we define the set $\calC(B)$
to be the set of all configurations for $B$ and we set $\calC\coloneqq\bigcup_{B\in\calB}\calC(B)$.
In our LP, for each configuration $C\in\calC$ we introduce a variable
$x_{C}$ representing whether we select $C$ for its corresponding
block. Then, we introduce constraints to model that we \fab{schedule} each
job at most once and we select one configuration for each block.
\begin{equation}
\tag{\textsf{LP}}
\begin{split}\max\sum_{C\in\calC}|J_C|\cdot x_{C}
\quad \text{s.t.} \quad \sum_{C\in\calC:j\in J_{C}}x_{C} & \leq1\quad\text{\ensuremath{\forall}}j\in J\\
\sum_{C\in\calC(B)}x_{C} & =1\quad\text{\ensuremath{\forall}}B\in\calB\\
x_{C} & \geq0\quad\text{\ensuremath{\forall}}C\in\calC.
\end{split}
\label{eq:configuration-lp}
\end{equation}
Due to Lemma~\ref{lem:good-block-superblock-partition}, the optimal objective function value of (\ref{eq:configuration-lp}) is \ale{at least close to} $|\OPT|$.
\begin{lem}\label{lem:valueLP}
Let $(\calB,\calS)$ be a block-superblock partition satisfying properties \ref{item:block-superblock-partition-lb}-\ref{item:fewBlocks} of Lemma~\ref{lem:good-block-superblock-partition}. Then the optimal objective function value of the associated configuration LP is at least \mbox{$(1-2\eps)|\OPT|$}.
\end{lem}

\begin{proof}
    Let $(x'_C)_{C\in \calC}$ be the integral solution of \eqref{eq:configuration-lp} which encodes $\OPT'$, \aw{i.e.,} for each configuration $C \in \calC$, set $x'_C = 1 $ if $J_C$ is the set of all the jobs in $\OPT'$ that are scheduled inside $B_{C}$, and set $x'_C = 0$ otherwise.
    Clearly, $(x'_C)_{C\in \calC}$ is a feasible solution to \eqref{eq:configuration-lp} and, by~\ref{item:block-superblock-partition-lb}, $(1-2\varepsilon)|\OPT| \leq |\OPT'| = \sum_{C\in \calC}x'_C |J_C|$ is a lower bound of the optimum of~\eqref{eq:configuration-lp}.
\end{proof}

Obviously, we can solve the above LP in polynomial time when $K$, i.e., the number of jobs per configuration, is upper bounded by a constant.
\begin{lem}\label{lem:solveLP}
Let $K$ denote the maximum number of jobs in a configuration. \fab{If $K=O_\eps(1)$,}
then in \fab{polynomial time} one can compute an optimal solution $(x^*_C)_{C\in \calC}$ to the configuration LP together with a feasible schedule $s_C$ for each configuration $C\in \calC$.
\end{lem}

\begin{proof}
There are at most $O(n)$ blocks and for each block $B$ there are at most $n^{O(K)}$ subsets $J'\aw{\subseteq J}$ of jobs of cardinality at most $K$ that can be feasibly scheduled inside $B$. For each such \aw{set} $J'$, we can check whether there exists an associated feasible schedule in time $K!\cdot n^{O(1)}$ by considering any ordering of the jobs in $J'$, and computing a greedy schedule (if any) of $J'$ respecting that order. Hence we can fully specify the configuration LP in time $n^{O(K)}$. This LP has $n^{O(K)}$ variables, $O(n)$ constraints besides the non-negativity ones, and \aw{integral} non-negative coefficients upper bounded by $K=O(n)$. Thus the LP can be solved in \aw{time $n^{O(K)}$}.
\end{proof}

\subsection{Rounding algorithm }\label{sec:firstRounding}

Let $(x_{C}^{*})_{C\in\calC}$ denote an optimal solution to \eqref{eq:configuration-lp}.
We describe now how to round it to an integral solution using randomized
rounding, losing at most a factor of $4/3+O(\eps)$ in the profit.
For each block $B\in\calB$ independently, we sample one configuration
$C^*(B)\in\calC(B)$ with respect to the distribution given by the vector
$(x_{C}^{*})_{C\in\calC(B)}$. More precisely, each configuration
$C\in\calC(B)$ is sampled with probability $x_{C}^{*}$ and, deterministically,
exactly one configuration in $\calC(B)$ is sampled. 

For each block $B\in\calB$ and each job $j\in J_{C^{*}(B)}$, there
is an interval $[s_{C^*(B)}(j),s_{C^*(B)}(j)+p_{j})$ during which $j$ is scheduled
by $s_{C^{*}(B)}$. We call this interval a \emph{slot} and we define
$Q$ to be the set of all slots for all blocks $B\in\calB$ and all
jobs $j\in J_{C^{*}(B)}$. Note that the slots in $Q$ are pairwise
disjoint. We want to assign a subset of the jobs in $J$ to the slots
in $Q$ via a bipartite matching. Possibly, this will assign a job
$j\in J$ to a slot $[s,t)$ corresponding to a different
job $j'\ne j$.

Formally, we define a bipartite graph $G=(V,E)$ where $V=J\dot{\cup}Q$ and there is an edge $\{j,[s,t)]\}\in E$ connecting (the vertices corresponding to) a job $j\in J$ and a slot $[s,t)\in Q$ 
if and only if $j$ can be scheduled during $[s,t)$; the latter conditions holds if and only if $\min\{t,d_j\}-\max\{s,r_j\}\geq p_j$, in which case $j$ can be scheduled during $[\max\{s,r_j\},\max\{s,r_j\}+p_j)$.
We compute a maximum matching $M^{*}\subseteq E$ in $G$. 
For each edge $e=(j,[s,t))\in M^{*}$
we select job $j$ and schedule it during $[\max\{s,r_j\},\max\{s,r_j\}+p_j)$. 
Finally, we output the resulting solution.

\subsection{Analysis}\label{sec:polynomialTimeAnalysis}

By construction, we schedule $|M^{*}|$ jobs in total. Hence, it remains
to show that $M^{*}$ is sufficiently large. To formalize this, we
call a job $j\in J$ \emph{local} if $\tw(j)\subseteq B$ for some
block $B\in\calB$ (note that then $B_{j,L}=B_{j,R}$) and \emph{global}
otherwise. We denote by $\Js$ and $\Jl$ the local and global jobs
in $J$, respectively. We prove that, compared to $(x_{C}^{*})_{C\in\calC}$,
we lose a factor of $4/3+O(\eps)$ in the profit of the global jobs
and a factor of $1+O(\eps)$ for the local ones. For convenience,
for each job $j\in J$ we define $y_{j}^{*}:=\sum_{C\in\calC:j\in J_{C}}x_{C}^{*}$
which is the total fractional extent to which $j$ is selected in
$(x_{C}^{*})_{C\in\calC}$. 
\begin{lem}
\label{lem:matchingExistence} For any $\Delta\geq 1$ we have
\[
\EE[|M^{*}|]\ge(1-O(\eps))\sum_{j\in\Js}y_{j}^{*}+(3/4-O(\varepsilon))\sum_{j\in\Jl}y_{j}^{*}\ge(3/4-O(\varepsilon))|\OPT|.
\]
\end{lem}
In the remainder of this section we prove Lemma~\ref{lem:matchingExistence}.
For this, we construct a feasible fractional bipartite matching, i.e., a function $f:E\rightarrow[0,1]$ such that $\sum_{e: v\in e}f(e)\leq 1$ for each $v\in V$.
We will prove that its \ale{expected} size $\ale{\EE}\left[\sum_{e\in E}f(e)\right]$ is at least the lower bound we want
to prove for $\ale{\EE}[|M^{*}|]$. Since the standard LP-relaxation of the bipartite matching problem is
integral (see \cite{schrijver2003combinatorial}), there exists also an integral matching with at least $\sum_{e\in E}f(e)$ edges.

Consider a job $j\in J$. If the sampled configurations $C^{*}(B_{j,L})$
for its release block $B_{j,L}$ contains $j$, then we match $j$
integrally to the corresponding slot. Formally, in this case let $[s,t)\fab{:=[s_{C^{*}(B_{j,L})}(j),s_{C^{*}(B_{j,L})}(j)+p_j)}$
denote the slot 
corresponding to $j$; we
define $f(\{j,[s,t)\}):=1$. If this is not the case but the sampled
configurations $C^{*}(B_{j,R})$ for its deadline block $B_{j,R}$
contains $j$, then, similarly, we match $j$ to the corresponding
slot. We do this operation for each job $j\in J$. Let $\Jbnd \subseteq J$
denote the set of all jobs that were \fab{(integrally)} matched in this way.

We want to define a fractional matching for the remaining jobs $J\setminus \Jbnd$.
We do this separately for each superblock $S\in\calS$. Let $J_{S}\subseteq J$ denote
all jobs in $J$ that span $S$. Note that this might include jobs
in $\Jbnd$. We want to define a fractional matching for the jobs in
$J_{S}\setminus \Jbnd$. For each job $j\in J_{S}$ we define the total
fractional amount by which $j$ is assigned to $S$ in $(x_{C}^{*})_{C\in\calC}$
by $y_{j,S}^{*}:=\sum_{B\in\calB(S)}\sum_{C\in\calC(B):j\in J_{C}}x_{C}^{*}$.
Similar to harmonic grouping \cite{lee1985simple,williamson2011design}, we partition the jobs in $J_{S}$
into $1/\eps$ groups $J_{1},...,J_{1/\eps}$ ordered non-increasingly by their processing
times such that the jobs from each group contribute almost the same
to the profit of $(x_{C}^{*})_{C\in\calC}$ within~$S$.

\begin{lem}
\label{lem:groups}There exists a partition $J_{1},...,J_{1/\eps}$
of $J_{S}$ such that 
\begin{enumerate}[label=(C\arabic*)]
\item for each $\ell\in[1/\eps-1]$, each job $j\in J_{\ell}$, and each
job $j'\in J_{\ell+1}$ we have that $p_{j}\ge p_{j'}$,
\item for each $\ell\in[1/\eps]$ we have that 
$
\eps\cdot\sum_{j\in J_{S}}y_{j,S}^{*}-1\le\sum_{j\in J_{\ell}}y_{j,S}^{*}\le1+\eps\cdot\sum_{j\in J_{S}}y_{j,S}^{*}.
$
\end{enumerate}
\end{lem}

\begin{proof}
    Assume that $J_S=\{j_1, \dots, j_k\}$ such that $p_{j_i}\geq p_{j_{i'}}$ for all $i, i'\in [k]$ with $i<i'$. Let $Y^*=\sum_{j\in J_S}y_{j, S}^*$ denote the total fractional amount of scheduled jobs from $J_S$ in $S$.
    For each $\ell\in [1/\eps]$ let $t(\ell)\in [k]$ be the maximal index $t$ such that $\sum_{i=1}^{t}y_{j_i, S}^*\leq \ell \cdot \eps Y^*$.
    Also, let $t(0):=0$.
    We define our partition by setting $J_{\ell}:=\{j_i\in J_S :\ale{t}(\ell-1)<i\leq \ale{t}(\ell)\}$.
    Clearly, this is indeed a partition of $J_S$.
    Also for $j_{i} \in J_\ell$, $j_{i'}\in J_{\ell +1}$ we have $i<i'$ and thus $p_{j_i}\geq p_{j_{i'}}$.

    Let $\ell\in [1/\eps]\cup\{0\}$. By definition of $t(\ell)$ and the fact that $y_{j, S}^*\in[0,1]$ for each $j\in J_S$ we have $\ell\cdot \eps Y^*-1\leq \sum_{i=1}^{t(\ell)}y_{j_i, S}^*\leq \ell\cdot \eps Y^*$.
    Thus we obtain 
    \begin{align*}
        \eps Y^*-1
        &=(\ell \cdot \eps Y^*-1 )-(\ell-1)\eps Y^*\\
        &\leq \sum_{i=1}^{t(\ell)}y_{j, S}^*-\sum_{i=1}^{t(\ell-1)}y_{j, S}^*\\
        &\leq \ell\cdot \eps Y^*-((\ell-1)\eps Y^*-1)\\
        &=\eps Y^*+1
    \end{align*}
    For each $\ell\in [1/\eps]$ we have $\sum_{j\in J_\ell}y_{j, S}^*=\sum_{i=1}^{t(\ell)}y_{j, S}^*-\sum_{i=1}^{t(\ell-1)}y_{j, S}^*$. 
    Altogether, this implies $\eps Y^*-1\leq \sum_{j\in J_\ell}y_{j, S}^*\leq \eps Y^*+1$, which completes the proof.
\end{proof}

For each $\ell\in[1/\eps]$ we define $n_{\ell}:=\sum_{j\in J_{\ell}}y_{j,S}^{*}$.
Note that Lemma~\ref{lem:groups} implies that $n_{\ell'}-2\le n_{\ell}\le n_{\ell'}+2$
for any $\ell,\ell'\in[1/\eps]$. Intuitively, when we define the
fractional matching for the jobs in $J_{S}\setminus \Jbnd$ we will ignore
the jobs in $J_{1}\setminus \Jbnd$ and for each $\ell\ge2$ we will
match the jobs in $J_{\ell}\setminus \Jbnd$ fractionally to the slots
corresponding to jobs in $J_{\ell-1}$. In particular, here we may
use slots that correspond to jobs in $\Jbnd$ \aw{(and to which we have not yet assigned any jobs)}. We would like that after
sampling the configurations for the blocks, for each group $J_{\ell}$
we obtain (essentially) $n_{\ell}$ slots in $S$ corresponding to
jobs in $J_{\ell}$. If $n_{\ell}$ is sufficiently large, we can
show that this is indeed the case with concentration arguments, using
that each superblock contains many blocks and the block's configurations
are sampled independently. On the other hand, if $n_{\ell}$ is small
then we can simply ignore the profit of jobs in $J_{S}$ within $S$: indeed, by Lemma \ref{lem:good-block-superblock-partition}, each superblock contains on average $\frac{K}{\eps^6}$ jobs.
Formally,
for each $\ell\in[1/\eps]$ we define $\bar{N}_{\ell}$ to be the
(random) number of slots corresponding to jobs in $J_{\ell}$ in the sampled
configurations $\{C^{*}(B)\}_{B\in\calB:B\subseteq S}$.

\begin{lem}
\label{lem:block-concentration}Let $\ell\in [1/\eps]$. With probability
at least $1-\eps^{\aw{2}}$ we have that
$\bar{N}_{\ell}\ge(1-\eps)n_{\ell}-2K/\eps^3$.
\end{lem}

\begin{proof}
Recall that $\calB(S):=\{B\in \calB: B\subseteq S\}$ denotes all blocks in the superblock $S$.
    For each $B\in \calB(S)$ let $Z_B$ denote the number of jobs from $J_\ell$ in the sampled configuration $C^*(B)$ and let $Z_B'=Z_B/K$.
    As there are at most $K$ jobs in a configuration, we have $Z_B'\in [0, 1]$.
    Note that $\bar{\fab{N}}_\ell=\sum_{B\in \calB_S}Z_B$ and $\EE[\bar{\fab{N}}_\ell]=n_\ell$.
    If $n_\ell< \frac{2K}{\eps^3}$, there is nothing to show. Hence in the following we assume $n_\ell\geq \frac{2K}{\eps^3}$ and thus $\EE[\sum_{B\in \calB_S}Z_B']\geq \frac{2}{\eps^3}$. \fab{It is sufficient to show that $\bar{\fab{N}}_\ell< (1-\eps)n_\ell$ happens with probability at most $\eps^{\aw{2}}$.} As the random variable $Z_B'$ depends only on the sampled configuration for the block $B$ and these configurations are sampled independently, the random variables $Z_B'$ for $B\in \calB_S$ are independent.
    Thus, we can apply Chernoff's bound \cite{dubhashi2009concentration} and obtain the following;
    \begin{align*}
        \Pr\left[\sum_{B\in \calB(S)}Z_B'<(1-\eps)\EE\left[\sum_{B\in \calB(S)}Z_B'\right]\right]
        &\leq \exp\left(-\frac{\eps^2}{2}\EE\left[\sum_{B\in \calB(S)}Z_B'\right]\right)\\
        &\leq \exp(-\frac{1}{\eps})\leq \eps^{\aw{2}}.
    \end{align*}
    As $\bar{\fab{N}}_\ell< (1-\eps)n_\ell$ is equivalent to $\sum_{B\in \calB(S)}Z_B'<(1-\eps)\EE\left[\sum_{B\in \calB(S)}Z_B'\right]$, this completes the proof.
\end{proof}

If the
\awr{commented out ``high-probability''}
event \aw{due to} Lemma~\ref{lem:block-concentration} does not
happen \aw{for some $\ell \in [1/\eps]$}, then we simply do not match the jobs in $J_{S}$ to the slots contained
in $S$ in our fractional matching. Since this happens only with probability
$\eps$, this influences our expected profit only marginally. Otherwise,
for each $\ell\in[1/\eps]$ we have essentially $n_{\ell}$ slots
available corresponding to the jobs in $J_{\ell}$. Therefore, since
$n_{\ell+1}\approx n_{\ell}$ we can (fractionally) match essentially
$n_{\ell+1}$ jobs from $J_{\ell+1}$ to these slots. Recall that
$\sum_{j\in J_{\ell+1}}y_{j,S}^{*}\fab{=} n_{\ell+1}$. Hence, we
could match each job $j\in J_{\ell+1}$ to a fractional extent of
$y_{j,S}^{*}$. Since the jobs in $J_{\ell+1}\cap \Jbnd$ are already
matched, we can even match each remaining job $j\in J_{\ell+1}\setminus \Jbnd$
to a larger fractional extent than only $y_{j,S}^{*}$. 

Let us define this increased extent. \aw{F}or each job $j\in J$,
let $y_{j,L}^{*}:=\sum_{C\in\calC(B_{j,L}):j\in J_{C}}x_{C}^{*}$
and $y_{j,R}^{*}:=\sum_{C\in\calC(B_{j,R}):j\in J_{C}}x_{C}^{*}$,
i.e., the probabilities that $j$ is contained in \fab{the sampled} configuration
for $B_{j,L}$ and $B_{j,R}$, resp. Hence, each job $j\in J_{\ell+1}$
is \emph{not} contained in $\Jbnd$ with probability $(1-y_{j,L}^{*})(1-y_{j,R}^{*})$.
Thus, in expectation, the sum of the values $y_{j,S}^{*}$ for the
(remaining) jobs in $J_{\ell+1}\setminus \Jbnd$ equals $\mathbb{E}\left[\sum_{j\in J_{\ell+1}\setminus \Jbnd}y_{j,S}^{*}\right]=\sum_{j\in J_{\ell+1}}y_{j,S}^{*}(1-y_{j,L}^{*})(1-y_{j,R}^{*})$.
Therefore, we try to match each job $j\in J_{\ell+1}\setminus \Jbnd$
even to an increased extent of $\frac{y_{j,S}^{*}}{(1-y_{j,L}^{*})(1-y_{j,R}^{*})}$.
In expectation, the sum of those values is then 
\[
\mathbb{E}\bigg[\sum_{j\in J_{\ell+1}\setminus \Jbnd}\frac{y_{j,S}^{*}}{(1-y_{j,L}^{*})(1-y_{j,R}^{*})}\bigg]=\sum_{j\in J_{\ell+1}}y_{j,S}^{*}= n_{\ell+1}.
\]

Via concentration arguments, we can argue that the sum of those values is also sufficiently concentrated around $n_{\ell+1}$, implying that we can find a matching of size close to $n_\ell$ in expectation. As mentioned in the introduction, the dependencies among the variables do not allow us to apply the standard Chernoff's bound. However, we are able to show that the impact of such dependencies is sufficiently small, hence we have sufficient concentration.
To define our matching formally,
let $Q_{S}\subseteq Q$ denote all slots in the blocks $\calB(S)$
corresponding to jobs in $J_{S}$ and let $E_{S}\subseteq E$ denote
all edges in $E$ that connect a vertex (corresponding to a job) in
$J_{S} \setminus \Jbnd$ with a slot in $Q_{S}$.
\begin{lem}
\label{lem:fractional-matching-exists} For each superblock $S\in \calS$, there exists a fractional matching  $f_{S}:E_{S}\rightarrow[0,1]$
such that
\begin{itemize}\itemsep0pt
\item $\EE[\sum_{e\in E_{S}}f_{S}(e)]\ge(1-O(\eps))\sum_{j\in J_{S}}y_{j, S}^*-2K/\eps^5$ and
\item for each $j\in J_S\setminus \Jbnd$ we have
\begin{equation*}
    \sum_{q: \{j,q\}\in E_S}f_S(\{j,q\})\leq \frac{y_{j,S}^{*}}{(1-y_{j,L}^{*})(1-y_{j,R}^{*})}\leq 1.
\end{equation*}
\end{itemize}
\end{lem}

\aw{We will prove Lemma~\ref{lem:fractional-matching-exists} in Section~\ref{sec:OmmitedProofTechnicalConcentration}.}
We combine all these matchings for the superblocks $S\in \calS$,
together with the values for $f$ that we defined already for the jobs in $\Jbnd$. Formally,
for each each superblock $S\in\calS$ and each edge $e\in E_{S}$
we define $f(e):=f_{S}(e)$. For each edge $e'\in E$ for which we
have not defined the value $f(e')$ yet, we set $f(e'):=0$. 

It remains to argue that $\sum_{e\in E}f(e)$ is sufficiently large
in expectation. For each job $j\in J$ we define the extent to which
$j$ is matched in $f$ by \emph{$g(j):=\sum_{q\in Q:\{j,q\}\in E}f(\{j,q\})$.
}We can easily show for each (local) job $j\in\Js$ that $g(j)=1$
with probability $y_{j}^{*}$. 
\begin{lem}
\label{lem:short-jobs-matched}For each job $j\in\Js$ we have that
$\Pr[g(j)=1]=y_{j}^{*}$.
\end{lem}

\begin{proof}
There is a unique block $B\in\calB$ containing $\tw(j)$. Hence,
$g(j)=1$ \fab{iff} for $B$ we sampled a configuration $C$
containing $j$. This happens with probability $\sum_{C\in\calC(B):j\in J_{C}}x_{C}^{*}=y_{j}^{*}$. 
\end{proof}
Consider a (global) job $j\in\Jl$. At the beginning, we matched $j$
to one slot in $B_{j,L}$ or $B_{j,R}$ with probability $y_{j,L}^{*}+y_{j,R}^{*}-y_{j,L}^{*}\cdot y_{j,R}^{*}$.
On the other hand, the LP-solution obtains a profit of $y_{j,L}^{*}+y_{j,R}^{*}$
from assigning $j$ fractionally to $B_{j,L}$ and $B_{j,R}$. However,
since $y_{j,L}^{*}+y_{j,R}^{*}\le1$ the latter profit is by at most
a factor of $4/3$ larger than the probability that we matched $j$
to some slot in $B_{j,L}$ or $B_{j,R}$.

\begin{prop}\label{prop:meanInequality}
For each job $j\in J$ we have that $y_{j,L}^{*}+y_{j,R}^{*}\le\frac{4}{3}(y_{j,L}^{*}+y_{j,R}^{*}-y_{j,L}^{*}\cdot y_{j,R}^{*})$.
\end{prop}

\begin{proof}
Note that $y_{j,L}^{*}\geq 0,y_{j,R}^{*}\geq 0$.
Due\fabr{Slightly shortened} to the arithmetic-geometric mean inequality we have $y^*_{j,L}\cdot y^*_{j,R}\leq \left(\frac{y^*_{j,L}+ y^*_{j,R}}{2}\right)^2\leq \frac{y^*_{j,L}+ y^*_{j,R}}{4}$, where in the last inequality we used $y_{j,L}^{*}+ y_{j,R}^{*}\leq 1$. Therefore
$y^*_{j,L}+y^*_{j,R}-y^*_{j,L}\cdot y^*_{j,R}\geq \frac{3}{4}(y^*_{j,L}+y^*_{j,R})$.
\end{proof}

Assume now that $j$ is not matched to a slot in $B_{j,L}$ nor in $B_{j,R}$, which
happens with probability $(1-y_{j,L}^{*})(1-y_{j,R}^{*})$. Roughly
speaking, for each superblock $S\in\calS$ spanned by $j$, we match
$j$ fractionally to a total extent of $(1-O(\eps))\frac{y_{j,S}^{*}}{(1-y_{j,L}^{*})(1-y_{j,R}^{*})}$
to the blocks in $\calB(S)$. More precisely, this holds on average
over all the global jobs $j$ spanning each superblock $S$. Using this,
we can prove the following bound for the fractionally matched global
jobs.
\begin{lem}
\label{lem:long-jobs-matched}For any $\Delta\ge 1$ we have that\newline
\makebox[0pt][l]{$\mathbb{E}\left[\sum_{j\in\Jl}g(j)\right]\ge\frac{3}{4}(1-O(\eps))\sum_{j\in\Jl}y_{j}^{*}-O(\eps)|OPT'|.$}
\end{lem}

\begin{proof}
    Recall that we fully matched the jobs in $\Jbnd$ and combined this with the matchings $f_S$ for each superblock $S\in \calS$. 
    This implies $$\EE\left[\sum_{j\in\Jl}g(j)\right]=E\left[|\Jbnd\cap \Jl|\right]+\sum_{S\in \calS}\EE\left[\sum_{e\in E_{S}}f_{S}(e)\right].$$
    By construction we have for each $j\in \Jl$ that $\Pr[j\in \Jbnd]=y_{j, L}^*+y_{j, R}^*-y_{j, L}^*\cdot y_{j, R}^*$.
    Using Proposition~\ref{prop:meanInequality} we obtain $\Pr[j\in \Jbnd] \geq \frac{3}{4}(y_{j, L}^*+y_{j, R}^*)$ and consequently $\EE[|\Jbnd\cap \Jl|]=3/4\cdot \sum_{j\in \Jl}(y_{j, L}^*+y_{j, R}^*)$.
    Now we combine this with Lemma~\ref{lem:fractional-matching-exists} and obtain
    \begin{align*}
        \mathbb{E}\left[\sum_{j\in\Jl}g(j)\right]
        &\geq \frac{3}{4} \sum_{j\in \Jl}(y_{j, L}^*+y_{j, R}^* )+ \sum_{S\in \calS}\left((1-O(\eps))\sum_{j\in J_{S}}y_{j, S}^*-2K/\eps^5\right)\\
        &\geq \frac{3}{4}(1-O(\eps))\sum_{j\in \Jl}\left(y_{j, L}^*+y_{j, R}^*+\sum_{S\in \calS:S\subseteq \tw(j)}y_{j, S}^*\right)-2K/\eps^5\cdot |\calS|\\
        &= \frac{3}{4}(1-O(\eps))\sum_{j\in \Jl}y_{j}^*-2K/\eps^5\cdot |\calS|
    \end{align*}
By Lemma~\ref{lem:good-block-superblock-partition}, \ref{item:fewBlocks}, we have 
$2K/\eps^5\cdot |\calS|\leq 2\eps|\OPT'| $.
Thus together with the above, we obtain
    $$\mathbb{E}\left[\sum_{j\in\Jl}g(j)\right]\ge\frac{3}{4}(1-O(\eps))\sum_{j\in\Jl}y_{j}^{*}-2\eps|\OPT'|$$
    This completes the proof.
\end{proof}

Now, Lemma~\ref{lem:matchingExistence} essentially follows Lemmas~\ref{lem:valueLP}, \ref{lem:short-jobs-matched},
\ref{lem:long-jobs-matched}, and the integrality of the bipartite
matching polytope.

 \begin{proof}[Proof of Lemma~\ref{lem:matchingExistence}]
By Lemmas~\ref{lem:short-jobs-matched} and
\ref{lem:long-jobs-matched}, there exists a fractional matching of expected size 
$\sum_{j\in\Js}y_{j}^{*}+\frac{3}{4}(1-O(\eps))\sum_{j\in\Jl}y_{j}^{*}-\eps|OPT'|$. 
By Lemma~\ref{lem:valueLP} we have $\sum_{j\in J}y_{j}^*\geq (1-2\eps)\OPT$, which implies that the fractional matching has an expected size of at least $$(1-O(\eps))\sum_{j\in\Js}y_{j}^{*}+(3/4-O(\varepsilon))\sum_{j\in\Jl}y_{j}^{*}\ge(3/4-O(\varepsilon))|\OPT|.$$
As the standard matching LP is integral on bipartite graphs, there is always an integral matching of size equal to the optimal LP value, which completes the proof.
\end{proof}

Theorem~\ref{thr:mainPoly} holds by choosing $\Delta=1$ and combining Lemmas~\ref{lem:good-block-superblock-partition}, \ref{lem:solveLP}, and \ref{lem:matchingExistence}.

\begin{proof}[Proof of Theorem~\ref{thr:mainPoly}]
    By Lemma~\ref{lem:matchingExistence}, the computed matching has a size of at least $\frac{3}{4}(1-O(\eps))\cdot |\OPT|$. By rescaling $\eps$ appropriately, we obtain an $\big(\frac{4}{3}+\eps\big)$-approximation algorithm.
    The maximum matching can be computed in time $n^{O(1)}$.
    The LP can be solved in time $n^{O(K)}$ by Lemma~\ref{lem:solveLP}.
    As we choose $\Delta=1$, we have $K=O_{\eps}(1)$ due to Lemma~\ref{lem:good-block-superblock-partition}.
    Thus the running time of the algorithm is bounded by $n^{O_{\eps}(1)}$.
\end{proof}

\subsection{Proof of Lemma~\ref{lem:fractional-matching-exists}}\label{sec:OmmitedProofTechnicalConcentration}
\fab{Let us focus on a specific superblock $S\in \calS$.} We would like to match a job $j$ to an extend of 
$$Y_{j}:=\begin{cases}
    \frac{y_{j, S}^*}{(1-y_{j, L}^*)(1-y_{j, R}^*)} &\text{ if } j \in J_{S}\setminus \Jbnd\\
    0 &\text{ if } j\in \Jbnd
\end{cases}$$ 
\fab{Consider the sets $J_\ell$ as defined in Lemma \ref{lem:groups}}, and
define $N_\ell:=\sum_{j\in J_\ell}Y_{j}$ for $\ell\in [1/\eps]\setminus \{1\}$.
First, we show some technical properties.
\begin{prop}\label{prop:hatn-properties}
    For each $j\in J_S$ we have (deterministically) $Y_{j}\leq 1$ and for each $\ell\in [1/\eps]$ we have $\EE[N_\ell]=n_\ell$.
\end{prop}
\begin{proof}
    Let $j\in J_S$.
    As $y_{j, L}^*+y_{j, R}^*+\sum_{S'\in \calS:S'\subseteq \tw(j)}y^*_{j, S'}\leq 1$ due to the constraints of \eqref{eq:configuration-lp}, we have $y_{j, S}^*\leq 1-y_{j, L}^*-y_{j, R}^*\leq (1-y_{j, L})(1-y_{j, R})$. This implies $Y_{j}\leq 1$.
    
    As $\Pr[j\in J_S \setminus \Jbnd]=(1-y_{j,L}^*)(1-y_{j, R}^*)$, we have $\EE[Y_{j}]=y_{j, S}^*$ Summing over all $j\in J_\ell$, we obtain $\EE[N_\ell]=\sum_{j\in J_\ell}y_{j, S}^*=n_\ell$.
\end{proof}
As already mentioned, we want to match a job $j\in J_\ell$ to an extend of $Y_{j}$ to slots corresponding to jobs in $J_{\ell-1}$.
As the time window of $j$ spans $S$ and each slot corresponding to a job in $J_{\ell-1}$ has a length of at least $p_j$ due to Lemma~\ref{lem:groups}, we can match $j$ to any slot corresponding to any job in $J_{\ell-1}$.
The total extend to which we want to match jobs is $N_\ell$ and the total number of slots is $\bar{\fab{N}}_{\ell-1}$.
As we have a complete bipartite graph between these sets of nodes, the maximum matching has a size of $\min\{N_\ell, \bar{\fab{N}}_{\ell-1}\}$.
Therefore, the major part of the proof is devoted to show that $\min\{N_\ell, \bar{\fab{N}}_{\ell-1}\}$ is in expectation close to $n_\ell$.

In the following paragraph, we give some intuition. 
By Lemma~\ref{lem:block-concentration} we already know that $\bar{\fab{N}}_{\ell-1}$ is close to $n_\ell$ with probability $1-\eps^2$.
So the goal is to obtain a similar bound for $N_\ell$. 
We know that $\EE[N_\ell]=n_\ell$ by Proposition~\ref{prop:hatn-properties}, so we only need some kind of concentration for $N_\ell$.
Recall that \fab{$Y_{j}$} depends on the sampled configurations for the boundary blocks for $j$.
This dependence also implies that the random variables $Y_{j}$ are not independent.
Hence we cannot directly apply standard Chernoff's bound to obtain the desired concentration for $N_\ell$.
To our advantage, the correlation between the random variables is still quite low: 
If two random variables $Y_{j}$ and $Y_{j'}$ are not independent, then there must be a block which is a boundary block for both $j$ and $j'$.
We make this dependence even weaker. We define a new random variable $\tilde{Y}_j$, which equals $Y_j$ except if there is a boundary block $B$ of $j$ with $\Pr[j\in J_{C^*(B)}]<\eps$ and nevertheless we sample a configuration $C^*(B)$ with $j\in J_{C^*(B)}$. 
In this case $\tilde{Y}_j$ behaves the same as $Y_j$ behaves when we sample a configuration $C^*(B)$ with $j\not\in J_{C^*(B)}$.
As this event only occurs with probability $\eps$ this sacrifices only an $\eps$ fraction of the profit in expectation. 
The advantage of this is that for a job $j$, there are now only $2K/\eps$ other jobs $j'$ for which $Y_{j}$ and $Y_{j'}$ are not independent. \fab{Indeed,} each such job $j'$ needs to share a boundary block with $j$, and \fab{$j'$} must be scheduled in this boundary block with probability at least $\eps$. \fab{Since at most $K$ fractional jobs can be scheduled in a block, there can be at most $K/\eps$ jobs fractionally assigned to a block by an amount of at least $\eps$.} This is enough independence to obtain the desired concentration for $N_\ell$, using the results by \cite{gavinsky2015tail}.

Now we make this formal.
First we introduce the concept of read-$k$ families (see \cite{gavinsky2015tail}), which are defined as follows. 
\begin{defn}[Read-$k$ families \cite{gavinsky2015tail}]\label{def:read-k-family}
    Let $C_1, \dots, C_b$ be independent random variables and let $k\in \NN$ and $J'$ be a finite set. For each $j \in J'$, let $A_j \subseteq [b]$ and let $f_j:(C_p)_{p\in A_j}\to [0,1]$. Assume that $|\{j:\;b'\in A_j\}| \leq k$ for every $b'\in [b]$.
    Then the random variables $Z_j = f_j\big((C_p)_{p\in A_j}\big)$ for $j\in J'$ are called a \emph{read-$k$ family}.
\end{defn}
Intuitively, \fab{in our setting} the random variables $C_1, \dots, C_b$ are the sampled configurations $C^*(B)$ for the blocks $B\in \calB$. For every job $j \in J_S$, the set $A_j$ is a subset of the boundary blocks for $j$, and  $Z_j$ is \fab{a random variable \emph{close to}} $\fab{Y_{j}}$.\fabr{Removed reference to function since $Z_j$ is enough} For each job $j\in J_S$, we would like to choose $\fab{A}_j$ as the two boundary blocks \fab{for $j$} and $Z_j$ as the random variable $Y_{j}$, as the value of $Y_{j}$ is determined by the sampled configurations for two boundary blocks \fab{for} $j$, i.e.$B_{j, L}$ and $B_{j, R}$. 
But this doesn't result in a read-$k$-family (for a reasonable value of $k$) as it might be that all jobs are released within the same block and thus all have the same boundary block and all depend on the same sampled configuration.
Therefore, we will define new random variables $\tilde{Y}_{j}$ with stronger independence properties (which will form the read-$k$-family) as follows:
\ale{$$\tilde{Y}_{j}:=\begin{cases}
    0 &\text{ if } y_{j, L}^*\geq \eps \text{ and } j \in J_{C^*(B_{j, L})}\\
    0 &\text{ if } y_{j, R}^*\geq \eps \text{ and } j \in J_{C^*(B_{j, R})}\\
    \frac{y_{j, S}^*}{(1-y_{j, L}^*)(1-y_{j, R}^*)} &\text{ otherwise }
\end{cases}$$ }
and let $\tilde{N}_\ell:=\sum_{j\in J_\ell}\tilde{Y}_{j}$. 
Using these random variables, we obtain the following properties.
\begin{lem}\label{lem:read-k-fulfilled}
The following holds.
\begin{itemize}
    \item For each $\ell\in [1/\eps]$ we have $\tilde{N}_\ell\geq N_\ell$ and $\EE[ \tilde{N}_\ell- N_\ell]\leq \ale{O(\eps)}\cdot\sum_{j\in J_\ell} y_{j, S}^*$.
    \item The random variables $\tilde{Y}_{j}$ for $j\in J_S$ are a read-$(K/\eps)$ family.
\end{itemize}
\end{lem}
\begin{proof}
    We start with the first property. Let $j\in J_S$. 
    By construction, the random variables $Y_{j}$ and $\tilde{Y}_{j}$ can only take the two values $0$ and $\frac{y_{j, S}^*}{(1-y_{j, L}^*)(1-y_{j, R}^*)}$.
    Note that $\tilde{Y}_{j}=0$ implies $j\in J_{C^*(B_{j, L})}$ or $j\in J_{C^*(B_{j, R})}$. Both imply $j\in \Jbnd$ and thus $Y_{j}=0$.
    Thus we obtain $\tilde{Y}_{j}\geq Y_{j}$ and by summing this over all $j\in J_\ell$ we obtain $\tilde{N}_\ell\geq N_\ell$.

    Let $j\in J_S$. 
    Suppose that $\tilde{Y}_{j}>Y_{j}$, i.e., $Y_{j}=0$ and $\tilde{Y}_{j}>0$. 
    This implies $j\in \Jbnd$ and thus $j \in J_{C^*(B_{j, L})}$ or $j \in J_{C^*(B_{j, R})}$. 
    The fact that $\tilde{Y}_j>0$ implies that there exists $X\in \{L , R\}$ such that $j \in J_{C^*(B_{j, X})}$ and $y_{j, X}^*<\eps$.
    So $\tilde{Y}_{j}>Y_{j}$ implies that there is $X\in \{L, R\}$ with $y_{j, X}^*<\eps$ and $j \in J_{C^*(B_{j, X})}$.
    
    \alnote{I rewrote the next part}
    \ale{We condition on the event that $\tilde{Y}_j>0$.
    A crucial observation in that this conditioning is independent from the sampled configuration $C^*(B_{j, L})$ (\fab{resp.,} $C^*(B_{j, R})$) if $y_{j, L}^*<\eps$ (\fab{resp.,} $y_{j, R}^*<\eps$).
    Thus for $X\in \{L, R\}$ with $y_{j, X}<\eps$ we have $\Pr[j \in J_{C^*(B_{j, X})}|\tilde{Y}_j>0]\leq y_{j, X}^*\leq \eps$.
    This implies 
    \begin{align*}
        \Pr[\tilde{Y}_j>Y_j|\tilde{Y}_j>0]&\leq \sum_{X\in\{L, R\}:y_{j, X}^*<\eps}\Pr[j \in J_{C^*(B_{j, R})}|\tilde{Y}_j>0] \leq 2\eps.
    \end{align*}
    Note that $\tilde{Y}_j=\frac{y_{j, S}^*}{(1-y_{j, L}^*)(1-y_{j, R}^*)}$ when $\tilde{Y}_j>0$.
    Thus
    \begin{align*}
        \EE[ \tilde{Y}_{j}-Y_{j}]&=\frac{y_{j, S}^*}{(1-y_{j, L}^*)(1-y_{j, R}^*)}\cdot \Pr[\tilde{Y}_{j}>Y_{j}]\\
        &=\frac{y_{j, S}^*}{(1-y_{j, L}^*)(1-y_{j, R}^*)}\cdot \Pr[\tilde{Y}_j>0]\cdot \Pr[\tilde{Y}_{j}>Y_{j}|\tilde{Y}_j>0]\\
        &\leq \frac{2 \eps \Pr[\tilde{Y}_j>0]}{(1-y_{j, L}^*)(1-y_{j, R}^*)}y_{j, S}^*\\
        &\leq \frac{2\eps}{(1-\eps)^2} y_{j, S}^*\\
        &=O(\eps)\cdot y_{j, S}^*.
    \end{align*}}
    Again, summing over all  $j\in J_\ell$ yields $\EE[ \tilde{N}_\ell- N_\ell]\leq O(\eps)\cdot\sum_{j\in J_\ell} y_{j, S}^*$.

    Now we show the second property.
    By construction $\tilde{Y}_{j}$ can only depend on the sampled configurations for the blocks $B_{j, L}$ and $B_{j, R}$ and not on any other block in $\calB$. 
    Let $X\in \{L, R\}$. Then $\tilde{Y}_{j}$ depends on the configuration of the block $B_{j, X}$ only if $y_{j, X}^*\geq \eps$.
    Consider a block $B\in \calB$. For $\tilde{Y}_{j}$ to be dependent on $C^*(B)$, the above implies that $j$ is scheduled to an extend of $\eps$ within $B$ by the optimal LP solution, i.e., $B=B_{j, L}$ and $y_{j, L}\geq \eps$ or $B=B_{j, R}$ and $y_{j, R}\geq \eps$. 
    Notice that at most $K$ jobs can be \fab{fractionally} scheduled within any block $B$. \fab{Indeed,}
    $$
    \fab{\sum_{j\in J}y^*_{j,B}=\sum_{j\in J}\sum_{C\in \calC(B):j\in J_C}x^*_{C}=\sum_{C\in \calC(B)}|J_C|x^*_C\leq K.}
    $$   
Therefore there can be at most $K/\eps$ jobs dependent on any block $B$.
    This implies that the random variables $\tilde{Y}_{j}$ for $j\in J_S$ are a read-$K/\eps$ family, by choosing $C_1, \dots, C_{\fab{|\calB|}}$ as the sampled configurations $C^*(B)$ for the blocks $B\in \calB$ and for each $j\in J_S$ the set $A_j$ contains the block $B\in \calB$ if and only if $B=B_{j, L}$ and $y_{j, L}\geq \eps$ or $B=B_{j, R}$ and $y_{j, R}\geq \eps$.
\end{proof}
As we now have a read-$K/\eps$ family of random variables, we can apply \cite[Theorem~1.1]{gavinsky2015tail}. This states that the adjustment of Chernoff's bound also holds for read-$k$ families, when the exponent is divided by $k$.
\begin{lem}[\cite{gavinsky2015tail}]\label{lem:read-k-concentration}
Let $Y_1, \dots ,Y_r$ be a family of read-$k$
     variables taking values in $[0, 1]$. Then for any $\delta>0$ we have
     $$\Pr\Big[\sum_{s=1}^{r}Y_s< (1-\delta) \EE\big[\sum_{s=1}^{r}Y_s\big]\Big]\leq \exp\left(-\frac{\delta^2}{2\cdot k}\cdot \EE\big[\sum_{s=1}^r Y_s\big]\right).$$
\end{lem}
\begin{proof}
    The statement in \cite{gavinsky2015tail} is slightly more general. Let $p=\EE\big[\sum_sY_s\big]/r$. 
    In \cite{gavinsky2015tail} it is shown that, for each $\eps'>0$,  $$\Pr\Big[\sum_{s=1}^r Y_s< (p-\eps')r\Big]\leq \exp\left(-D_{KL}(p-\eps'||p)\cdot \frac{r}{k}\right)$$ where $D_{KL}(q||p):=q\ln(\frac{q}{p})+(1-q)\ln(\frac{1-q}{1-p})$ is the Kullback-Leibler divergence. 
    We use the standard lower bound\alnote{I added a proof, because I didn't find a good source to cite this inequality...} 
    \begin{equation}\label{eqn:KLlowerBound}
    D_{KL}(q||p)\geq \frac{(p-q)^2}{2\cdot p}\quad \text{for }\fab{0<}q\leq p\fab{<1}
    \end{equation}
    for obtaining a multiplicative Chernoff's bound. We prove \eqref{eqn:KLlowerBound} later \fab{for completeness}. Let us first complete the proof using this lower bound.
    We choose $\eps'=\delta\cdot p$ and obtain $$D_{KL}(p-\eps'||p)\cdot r\geq \frac{(\delta p)^2}{2 p}\cdot r =\frac{\delta^2}{2}\EE\big[\sum_{s=1}^r Y_s\big]$$ which yields the desired result.
    
    Thus it remains to prove \eqref{eqn:KLlowerBound}.
    Consider the function $f:(0, \infty)\to \RR$, $f(x)=x\cdot \ln x$.
    Using the fundamental theorem of calculus, we obtain some bounds in $f(x)$.
    We have $f'(x)=1+\ln x$ and $f''(x)=1/x$.
    For $x\geq 1$, we have $f'(x)\geq 1$, which implies for all $x\geq 1$
    $$f(x)=f(1)+\int_{t=1}^{x}f'(t)\,dt\geq f(1)+\int_{t=1}^{x}1\,dt=x-1$$
    Similarly, for $x\leq 1$ we have $f''(x)\geq 1=f''(1)$, which implies for $x\leq 1$:\alnote{I changed it, such that the lower bound of the integral is smaller}
    \begin{align*}
        f(x)= f(1)\ale{-\int_{t=x}^1} f'(t)dt&=f(1)\ale{-\int_{t=x}^1} \ale{\Big(}f'(1)\ale{-\int_{s=t}^1} f''(s)\, ds \ale{\Big)} dt\\
        &\geq f(1)-f'(1)(1-x)+\frac{f''(1)}{2}(x-1)^2=(x-1)+\frac{1}{2}(x-1)^2.
    \end{align*}
    We can rewrite the Kullback-Leibler divergence using the function $f$:
    $$D_{KL}(q||p)=p\cdot f\Big(\frac{q}{p}\Big)+(1-p)f\Big(\frac{1-q}{1-p}\Big)$$
    As we assumed $q\leq p$, we have $\frac{q}{p} \leq 1$ and $\frac{1-q}{1-p}\geq1$.
    Now we plug in the above bounds and obtain
    \begin{align*}
        D_{KL}(q||p)&=p\cdot f\Big(\frac{q}{p}\Big)+(1-p)f\Big(\frac{1-q}{1-p}\Big)\\
        &\geq p\cdot \left(\frac{q}{p}-1+\frac{1}{2}\left(\frac{q}{p}-1\right)^2\right)+(1-p)\cdot \left(\frac{1-q}{1-p}-1\right)\\
        &=q-p+\frac{(q-p)^2}{2p}+(1-q)-(1-p)\\
        &=\frac{(q-p)^2}{2p}.
    \end{align*}
    This completes the proof.
\end{proof}
Now we apply the above concentration bound to $\tilde{N}_\ell$ and obtain the following result.
\begin{lem}\label{lem:holes-cocentration}
    Let $\ell\in [1/\eps]$. We have that
    \begin{equation*}
        \Pr\big[ \tilde{N}_\ell \leq (1-\eps) n_\ell-2 K/\eps^4 \big]\leq \eps.
    \end{equation*}
\end{lem}
\begin{proof}
    If $n_\ell\leq 2 K/\eps^4$, there is nothing to show. So assume $n_\ell\geq 2K/\eps^4$. \fab{It is sufficient to show that $\tilde{N}_\ell \leq (1-\eps) n_\ell$ happens with probability at most $\eps$.}
    By Lemma~\ref{lem:read-k-fulfilled}, the variables $\tilde{Y}_{j}$ are a read-$K/\eps$ family. 
    Also by Lemma~\ref{lem:read-k-fulfilled} and Proposition~\ref{prop:hatn-properties}, we have $\EE[\tilde{N}_\ell]\geq \EE[N_\ell]=n_\ell\geq 2 K/\eps^4$.
    By applying Lemma~\ref{lem:read-k-concentration} to $\tilde{N}_\ell$ for $\delta=\eps$, we obtain:
    $$ \Pr\big[ \tilde{N}_\ell \leq (1-\eps) n_\ell\big]\leq\exp\left( -\frac{\eps^2}{2\cdot K/\eps}\cdot 2K/\eps^4\right)\leq \exp(-1/\eps)\leq \eps.$$
\end{proof}
Now we can prove Lemma~\ref{lem:fractional-matching-exists}.
\begin{proof}[Proof of Lemma~\ref{lem:fractional-matching-exists}]

 For each $\ell\in [1/\eps]\setminus\{1\}$, let $f^\ell:E_S\to [0,1]$ denote a maximal \fab{fractional} matching between the jobs in $J_\ell$ and slots corresponding to jobs in $J_{\ell-1}$, where each job $j\in J_\ell$ is can be matched only to an extend of $Y_{j}$, i.e., which fulfills:
 \begin{itemize}
     \item For each $j\in J_S\setminus J_\ell$ and each $q\in Q$ we have $f^\ell(\{j, q\})=0$.
     \item For each $j\in J_\ell$ and each slot $q\in Q_S$ not corresponding to a job in $J_{\ell-1}$ we also have $f^\ell(\{j, q\})=0$.
     \item For each $j\in J^\ell$ we have $\sum_{q, \in Q_S}f^\ell(\{j, q\})\leq Y_{j}$.
 \end{itemize}
 Let $f_S(e):=\sum_{\ell=2}^{1/\eps}f_\ell(e)$ denote the union of these fractional matchings.
As shown before,\fabr{Before where? The proof is long. Maybe add and equation} we have $\sum_{e\in E_S}f_\ell(e)=\min\{N_{\ell}, \bar{\fab{N}}_{\ell-1}\} $.
 For each $j\in \Jbnd$ we have that $Y_{j}=0$ and thus $j$ is not matched at all.
And for each $j\in J_S \setminus \Jbnd$ there exists at most one $\ell\in [1/\eps]$ with $j\in J_\ell$ and thus $j$ is matched to an extend of at most $\frac{y_{j, S}^*}{(1-y_{j, L}^*)(1-y_{j, R}^*)}$.

It remains to compute the expected size of the matching. 
By construction, we have that 
$\sum_{e\in E}f_S(e)=\sum_{\ell=2}^{1/\eps}\min\{N_{\ell}, \bar{\fab{N}}_{\ell-1}\}$. Let $\ell\in [1/\eps]\setminus \{1\}$.
By Lemma~\ref{lem:block-concentration}, we have that $\bar{\fab{N}}_{\ell-1}\geq (1-\eps)n_{\ell-1}-2K/\eps^3$ holds with probability at least $1-\eps\ale{^2}$.
And by Lemma~\ref{lem:holes-cocentration}, we have that $\tilde{N}_{\ell}\geq (1-\eps){n_\ell}-2K/\eps^4$ holds also with probability at least $1-\eps$.
Thus with probability at least $1-2\eps$ both events occur. Thus we have with probability $1-2\eps$ that 
\begin{align*}
    \min\{\tilde{N}_\ell, \bar{\fab{N}}_{\ell-1}\}&\geq (1-\eps)\min\{n_\ell, n_{\ell-1}\}-\max\{2K/\eps^4,2K/\eps^3\}\\
    &= (1-\eps)\min\{n_\ell, n_{\ell-1}\}-2K/\eps^4
\end{align*}
Recall that $n_\ell\geq \eps\sum_{j \in J_S}y_{j, S}^*-1, n_{\ell-1}\geq \eps\sum_{j \in J_S}y_{j, S}^*-1$ by Lemma~\ref{lem:groups}.
Thus we obtain
\begin{align*}
    &\EE[\min\{\tilde{N}_\ell, \bar{\fab{N}}_{\ell-1}\}]\\
    &\geq \Pr\big[\min\{\tilde{N}_\ell, \bar{\fab{N}}_{\ell-1}\}\geq (1-\eps)\min\{n_\ell, n_{\ell-1}\}-2K/\eps^4\big] \\
    & \quad\ \cdot\left((1-\eps)\min\{n_\ell, n_{\ell-1}\}-2K/\eps^4\}\right)\\
    &\geq (1-2\eps)(1-\eps)(\eps\sum_{j \in J_S}y_{j, S}^*-1)-(1-\fab{2}\eps)\cdot 2K/\eps^4\\
    &\geq (1-3\eps)\cdot \eps \cdot \sum_{j \in J_S}y_{j, S}^* - 2K/\eps^4.
\end{align*}
This implies
\begin{align*}
    \EE[\sum_{e\in E}f_S(e)]
    &=\sum_{\ell=2}^{1/\eps}\EE[\min\{N_{\ell}, \bar{\fab{N}}_{\ell-1}\}]\\
    &\geq \sum_{\ell=2}^{1/\eps}\big(\EE[\min\{\tilde{N}_{\ell}, \bar{\fab{N}}_{\ell-1}\}]-\ale{\EE[\tilde{N}_\ell-N_\ell]\big)}\\
    &\overset{\text{Lem. }\ref{lem:read-k-fulfilled}}{\geq}(1-3\eps)(1-\eps) \sum_{j \in J_S}y_{j, S}^* - 2K/\eps^5\ale{-O(\eps)\sum_{j\in J_S}y_{j, S}^*}\\
    &\geq (1-O(\eps))\sum_{j \in J_S}y_{j, S}^* - 2K/\eps^5,
\end{align*}
which completes the proof.
\end{proof}

\section{A pseudopolynomial time $(5/4+\varepsilon)$-approximation}
\label{sec:pseudopolynomialTimeAlgorithm}

In this section we show how to improve our approximation ratio to
$5/4+\eps$, using pseudopolynomial running time.
Lemma~\ref{lem:matchingExistence}
states that in our previous algorithm, we lose a factor of $4/3+O(\eps)$
on the profit of the global jobs but only
a factor of $1+O(\eps)$ on the profit of the local jobs (compared to the optimal LP solution).
In this section, we present a different rounding
method that loses only a factor of $1+O(\eps)$ on the profit of the
global jobs, but does not schedule any local job. Therefore, the best of the two
solutions will then yield a $(5/4+O(\eps))$-approximation by a simple computation.

For this section, we define $\Delta:=4(\log_2 T+1)/\eps^4$. We want to apply Lemma~\ref{lem:good-block-superblock-partition} and solve the configuration-LP
for the resulting block-superblock \aw{partition}. However, it is not clear how to solve it in
(pseudo-)polynomial time since possibly $K=\Theta_{\eps}(\log T)$
and Lemma~\ref{lem:solveLP} \aw{guarantees polynomial running time only if $K=O_\eps(1)$}.
Also, it is not clear how to solve instances where $|\OPT| \le 5\Delta\fab{K_0}(2K_0/\eps^5)^{1/\eps}$ (see Lemma~\ref{lem:good-block-superblock-partition}) since the latter quantity can be up to $\Theta_{\eps}(\log T)$. \fab{This is discussed in Section \ref{sec:pseudopolyLP}. The alternative rounding algorithm is then presented in Section~\ref{sec:secondRounding}.}

\subsection{Solving the linear program and instances with small optimal solutions}
\label{sec:pseudopolyLP}

We resolve \fab{the above mentioned} issues \fab{related \aw{to our choice of} $\Delta=\Theta_\eps(\log T)$} with an algorithmic technique that, intuitively, searches for solutions with only few jobs and computes the best solution of this kind, even in a weighted generalization of our problem.
First, we show that we can solve \eqref{eq:configuration-lp}
in time $2^{O(K)}\cdot(nT)^{O(1)}$ in this way. To do this, we invoke the
ellipsoid method together with a suitable separation oracle for its
dual LP:\awr{todo: align $\forall$ statements \ant{done}}

\begin{alignat}{3}
    \min  \sum_{j\in J}\alpha_{j}+\sum_{B\in\calB}\beta_{B}  
    \quad \text{s.t.} \quad  \sum_{j\in J_C}\alpha_{j}+\beta_{B} & \geq|J_C| \quad && \forall  B\in \calB, C\in \calC(B) \tag{\textsf{dualLP}}\label{eq:configuration-lp-dual} \\
    \alpha_{j},\beta_B & \geq0 && \forall  j\in J, B\in \calB. \notag
\end{alignat}

The non-trivial task of the separation oracle is to determine, for a given tentative solution $(\alpha_{j},\beta_B)_{j\in J,B\in \calB}$ and a given block $B=[s,t)\in \calB$, whether there exists a configuration
$C\in\calC(B)$ with $\sum_{j\in J_C}\alpha_{j}+\beta_{B}<|J_C|$. This task is equivalent to computing a configuration $C^*=(B,J^*)\in\calC(B)$ of maximum weight, where the weight of each job $j$ is set to $w(j):=1-\alpha_j$. We show how to solve this task
in time $2^{O(K)}\cdot(nT)^{O(1)}$ using the color coding technique \cite{AYZ95}. We enumerate over a set of $K$-colorings of the jobs such that at least one such coloring assigns a distinct color to each \aw{job} $j\in J^*$. Then the problem reduces to computing a maximum-weight configuration that selects at most one job per color. We solve the latter task via dynamic programming. The idea is to construct a table indexed by a subset $COL\subseteq [K]$ of colors and an integer time $q\in \{s,\ldots,t\}$. The associated value corresponds to a maximum-weight subset of jobs, each one with a distinct color from $COL$, that can be feasibly scheduled within $[s,q)$. Each table entry can be computed in time $O(Kn)$ \aw{using previously computed entries for certain other entries}. The table entry with $(COL,q)=([K],t)$ corresponds to the desired solution. With a similar dynamic program, we can optimally solve instances
in the unweighted case\fabr{Actually also weighted: shall we mention this? Otherwise remove "in the unweighted case" which is implicit in the name for us \ale{I would keep it here, as we discuss directly before that the separation problem is the weighted version of the our problem and how to solve it}} of Maximum Throughput in time $2^{O(|\OPT|)}(nT)^{O(1)}$ which is pseudopolynomial as long as $|\OPT| = O_{\eps}(\log T)$.
\begin{lem}
\label{lem:solve-large-LP}In time $2^{O(K)}\cdot(nT)^{O(1)}$, we
can compute an optimal solution $(x_{C}^{*})_{C\in \calC}$ to \eqref{eq:configuration-lp}
together with a schedule $s_{C}$ for each configuration $C\in\calC$ with $x_{C}^{*}>0$. Also, Maximum Throughput can be solved exactly in time $2^{O(|\OPT|)}(nT)^{O(1)}$.
\end{lem}

\begin{proof}
    We start with the part of the claim related to the configuration LP. As already stated, the dual of \eqref{eq:configuration-lp} is \eqref{eq:configuration-lp-dual}.
Note that \eqref{eq:configuration-lp-dual} has $|J|+|\calB|=O(n)$ variables, but does not necessarily have a polynomial number of constraints. By the equivalence between optimization and separation (see, e.g., \cite[Section~14.2]{schrijver1998theory}), it is sufficient to present a separation oracle for the dual LP. In more detail, we are given values
$(\alpha_{j},\beta_B)_{j\in J,B\in \calB}$, and we need to check whether they induce a feasible
solution. We can easily check the non-negativity constraints. We next show how to check the remaining constraints for a fixed block $B\in \calB$. The separation oracle has to determine whether there exists a violated constraint, i.e., a configuration $C\in \calC(B)$ with $\sum_{j\in C} \alpha_j + \beta_{B} < |J_C|$. 
This is equivalent to $\sum_{j\in J_C}  (1-\alpha_j) > \beta_{B}$.
Let $w(j):=1-\alpha_j$ be the weight of a job. Then it suffices to compute a configuration $C^*=(B,J^*)\in \calC(B)$ such that $J^*$ has maximum weight, as such a configuration satisfies $\sum_{j\in J^*}  (1-\alpha_j) > \beta_{B}$ if and only if any configuration in $\calC(B)$ does. 

Now we show how to compute one such $C^*=(B,J^*)$ in pseudopolynomial time. This can be achieved by using the color-coding technique introduced in \cite{AYZ95}. 
We want to assign colors to jobs, such that all jobs in $C^*$ have different colors. 
This is crucial for the later dynamic program to keep track of which jobs we already scheduled.
Recall that each configuration fulfills $|C(B)|\leq K$. 
If we assign each job randomly one of $ K$ colors, there is a chance of $\frac{K!}{K^K}=1/2^{O( K)}$ that all jobs in $J^*$ have pairwise different colors. 
Intuitively, by trying $2^{O(K)}\cdot n^{O(1)}$ random colorings, we have a high probability that we have at least one coloring, such that all jobs in $J^*$ have pairwise different colors.

As observed in \cite{AYZ95}, this approach can also be de-randomized using a result\footnote{Even better results exist, but they are irrelevant for our goals.} in \cite{schmidt1990HashFunctions}.
For that, we use a so-called a family of $k$-perfect hash function. For $k\in \NN$ and a set $J$ of $n$ elements, a family of $k$-perfect hash functions over $J$ is a family of maps $f_i:J\mapsto [k]$, such that for each subset $A\subseteq J$ of size at most $k$, there exists a map $f_i$ such that $f_i$ is injective when restricted to $A$.  

\aw{In \cite{schmidt1990HashFunctions} it was shown that there exists a $k$-perfect family $\calF$ of hash functions over a set of $n$ elements of size $2^{O(k)}\log^{O(1)} n$. Furthermore, such a family can be computed in time $2^{O(k)} n^{O(1)}$. Using this,}
we compute such a $k$-perfect family of hash functions $\calF$ of size at most $2^{O( K)}\log^{O(1)} n$. We also refer to $\calF$ as the colorings. For each coloring $f_i\in \calF$, we will compute the maximum weight configuration in which all jobs have different colors. Given this, we output the maximum weight configuration among all colorings in $\calF$. As $J^*$ contains at most $ K$ jobs, there is a hash function $f_i$ such that all jobs in $J^*$ have different colors.

So it remains to show how we compute the maximum weight configuration $C^*=(B,J^*)$ for a block $B=[s,t)$ in which all jobs $J^*$ have different colors for a given coloring $f_i\in \calF$.
For that, we use a dynamic program.
Recall that a configuration for $B=[s,t)$ can only schedule jobs $j\in J$ with $r_{j}\in [s,t)$, or $d_{j}\in (s,t]$, or $j$ spans the superblock $S\in\calS$ with $B\subseteq S$.
Thus let $J'$ denote these jobs.
Let us define a DP-cell for each subset $A\subseteq [K]$ of colors and each integer $q\in \{s,\ldots,t\}$. 
The number of DP-cells is $O(2^KT)$. The value of the cell $(A,q)$ is a maximum weight set of jobs $\OPT(A, q)\subseteq J'$, which have the colors in $A$ once each and that can be scheduled inside $[s, q)$. We fill-in the table entries for increasing values of $|A|$ and increasing values of $q$. 
For $A=\emptyset$ or $q=s$, we set $\OPT(A, q)=\emptyset$. Otherwise, if $\OPT(A, q)$ is non-empty, there is a job $j\in \OPT(A, q)$ that can be scheduled last. This leads to the following recursion:
\begin{align}\label{eq:COloringDPrecursion}
    w(\OPT(A, q))&=\max\{0,\max_{j\in J': f_i(j)\in A} \max_{q'\in [s, q): |\tw(j)\cap [q', q)|\geq p_j} w(\OPT(A\setminus\{f_i(j)\}, q')\cup \{j\})\}
\end{align}
We assign to the cell $(A,q)$ the value $\OPT(A,q)=\emptyset$ if the above maximum is $0$, and otherwise the value $\OPT(A\setminus\{f_i(j)\}, q')\cup \{j\}$ attaining the maximum. We also compute a schedule associated to $\OPT(A, q)$ as follows. If the maximum above is $0$, we use the empty schedule. 
Otherwise, we set $s_{C^*}(j):=\max\{q', r_j\}$ for $j\in J', q'\in [s,q)$ attaining the maximum and combine this with the schedule associated to $(A\setminus\{f_i(j)\}, q')$, so that we obtain a valid schedule of weight $w(\OPT(A, q))$.
The above computation takes  $O(Kn)$ time per cell, hence $O(2^K TKn)$ time in total. Finally we set $J^*=\OPT([K], t)$. The associated schedule $s_{C^*}$ of $J^*\subseteq J'$ and the fact that obviously $|J^*|\leq K$ shows that $C^*$ is indeed a valid configuration. This yields the separation oracle.

Thus we can use the ellipsoid method to compute an optimal solution for \eqref{eq:configuration-lp-dual} using $n^{O(1)}$ calls to the separation oracle (since we have $O(n)$ variables and integer coefficients of absolute value $O(n)$), hence in total time $2^{O(K)}(nT)^{O(1)}$ (see, e.g., Theorem 6.4.9 in \cite{grotschel2012geometric}).

Next, we show how to use this to compute an optimal solution for the primal. 
Suppose that instead of the dual LP with a constraint for every configuration, we have a \ale{modified} dual LP with only the $n^{O(1)}$ constraints returned by the separation oracle.
Then we still would obtain the same optimal solution for the modified dual LP, as the algorithm behaves the same for both the original and the modified dual LP. 
Consider the optimal solution for the primal of the modified dual LP. 
This solution also has to be an optimal solution for \eqref{eq:configuration-lp} as it is feasible and has the same cost.
Also, it is $0$ for all variables corresponding to constraint not returned by the $n^{O(1)}$ calls to the separation oracle.
We can therefore optimally solve the primal LP restricted to the latter variables $x_C$ in polynomial time.

Given a configuration $C$ with $x^*_C>0$, we can compute a corresponding feasible schedule $s_C$ in time $2^{O(K)}(nT)^{O(1)}$ by using essentially the same type of dynamic program that we used in the separation oracle.

For the final part of the claim, we use the color-coding technique similarly to design of the separation oracle. In more detail, we compute a set of $2^{O(|\OPT|)}\log^{O(1)}n$ colorings of the jobs such that at least one such coloring assigns a distinct color to each $j\in \OPT$. Then the problem reduces to finding a feasible schedule of a maximum cardinality set of jobs such that at most one job per color is selected. This problem can be solved via a dynamic program similar to the one used in the mentioned separation oracle, with the difference that now all the job weights are $1$, $K$ is replaced by $|OPT|$ and that the block $B=[s,t)$ is replaced by the entire time horizon $[0,T)$. The overall running time is $2^{O(|\OPT|)}(nT)^{O(1)}$ as claimed.
\end{proof}

\subsection{Alternative rounding algorithm}\label{sec:secondRounding}


Due to Lemma~\ref{lem:solve-large-LP}, we can optimally solve all instances for which $|\OPT| < 5\ale{K_0}\Delta(2K_0/\eps^5)^{1/\eps}$.\fabr{Like before, we might need to add a factor $K_0$ \ale{Yes, I added the factor}} If this is not the case, we apply Lemma \ref{lem:good-block-superblock-partition} with our defined value of $\Delta$ to compute a block-superblock \aw{partition}\fabr{Don't we use "partition" before? I actually like decomposition more, but we need to be consistent}\awr{changed it to be consistent} satisfying properties \ref{item:block-superblock-partition-lb}-\ref{item:fewBlocks} from Lemma \ref{lem:good-block-superblock-partition}. %
Based on that we define~\eqref{eq:configuration-lp}.
Via Lemma~\ref{lem:solve-large-LP}, we compute an optimal solution $(x_{C}^{*})_{C\in\calC}$ for it in time $(nT)^{O_{\eps}(1)}$, using that $K=O_\eps(\log T)$.
In the following, we present now an alternative LP-rounding
algorithm which is complementary to the rounding algorithm from Section~\ref{sec:firstRounding}.
In the latter one, we sampled a configuration for each block and then
we assigned the jobs to the blocks. Instead, now we randomly assign each
job to a block and then try to find a schedule within each block,
possibly by discarding some of its assigned jobs.

We ignore all the (local) jobs $\Js$ and do not schedule any of them. For each block $B$ and job $j\in \Jl$, let $y^*_{j,B}:=\sum_{C\in\calC(B):j\in J_C} x_{C}^{*}$ be the fractional amount by which $j$ is assigned to block $B$ in the optimal LP solution. 
For each job $j\in\Jl$ independently, we do the following.
We randomly decide to assign $j$  to some block $B\in \calB$ or to discard $j$ such that $j$ is assigned to each block $B\in \calB$ with probability $(1-2\varepsilon)y^*_{j,B}$
and, deterministically, $j$ is assigned to at most one block in $\calB$. Therefore, $j$ is not assigned to any block and, hence, discarded with probability
$1-(1-2\eps)\sum_{B\in \calB}y^*_{j,B}$.
For each block $B\in\calB$ we define a random variable
$Y_{j,B}\in\{0,1\}$, modeling whether we assign $j$ to $B$ or not, such
that $\Pr[Y_{j,B}=1]=(1-2\varepsilon)y^*_{j,B}$. Notice that
deterministically $\sum_{B\in\calB}Y_{j,B}\le 1$. 

Consider a block $B=[s,t)\in\calB$ and the jobs $J(B)\subseteq \Jl$ randomly assigned to $B$,
i.e., all jobs $j$ such that $Y_{j,B}=1$. In a second phase, we
discard some of these jobs $J(B)$ such that the remaining jobs can be scheduled
within $B$, similar to randomized rounding with alteration, see, e.g., \cite{Srinivasan2001, chakrabarti2007approximation}.
Formally, we define a random variable $Z_{j,B}\in\{0,1\}$ for all
jobs $j\in J(B)$, where $Z_{j,B}=0$ indicates
that we discard $j$ (if $Y_{j,B}=0$, simply define $Z_{j,B}=0$). 

First, we discard all
jobs $j$ whose processing time $p_{j}$ is relatively long compared
to $\tw(j)\cap B$, i.e., the portion of $\tw(j)$ inside $B$. Formally,
for each interval $I=[s',t')$, we define its length by $|I|:=t'-s'$.
The set of \emph{long} jobs is given by $\Jt(B):=\{j\in J(B): p_{j}>\eps^4|\tw(j)\cap B|\}$. We can show that in expectation $\Jt(B)$ contains at most $\Delta$ jobs, which we can afford to discard since on average $(x_{C}^{*})_{C\in\calC}$
schedules at least $\Delta/\eps$ jobs in each block.
\begin{lem}\label{lem:pseudoRemovalLongJobs}
For each block $B\in\calB$ we have that $\E[|\Jt(B)|]\le 4(\log_2 T+1)/\eps^4=\Delta$. 
\end{lem}

\begin{proof}
    Fix a block $B=[s,t)\in \calB$ and let $\ell\in \NN$. Consider the jobs 
    $$
    J_B^\ell:=\{j\in J:2^{\ell}\leq p_j<2^{\ell+1}, \quad p_j> \eps^4 |\tw(j)\cap B|\}.
    $$ 
    There are no such jobs for $\ell > \log_2 T$.
    Note that each long job $j\in \Jt(B)$ belongs to $J_B^\ell$ for some value $\ell$. So suppose that $0\leq \ell \leq \log_2T$.

    Consider a job $j\in J_B^\ell$. If such a job is schedule within $B$, it has to be within the interval $\tw(j)\cap B$.
    Note that $|\tw(j)\cap B|<p_j/\eps^4\leq 2^{\ell+1}/\varepsilon^4$. 
    Also, we have $s\in \tw(j)\cap B$ or $t\in \tw(j)\cap B$ as we discarded all local jobs. Consequently, this implies $\tw(j)\cap B\subseteq [s, s+2^{\ell+1}/\eps^4)\cup [t-2^{\ell+1}/\eps^4, t)$. 
    Note that the latter has a total length of $2\cdot 2^{\ell+1}/\eps^4$. 
    As each job $j\in J_B^\ell$ has a processing time of at least $2^\ell$, any configuration for $B$ can contain at most $4/\eps^4$ such jobs.
    As the following calculation shows, the set $\Jt(B)$ contains at most $4/\eps^4$ jobs from $J_B^\ell$ in expectation.
    \begin{align*}
        \EE[|\Jt(B)\cap J_B^\ell|]&=\EE\left[ \sum_{j\in \Jl\cap J_B^\ell}Y_{j, B} \right] \\
        &=\sum_{j\in \Jl\cap J_B^\ell}\EE[Y_{j, B}]\\
        &\leq \sum_{j\in \Jl\cap J_B^\ell} \sum_{C\in\calC(B):j\in C}x_{C}^{*}\\
        &\leq  \sum_{C\in\calC(B)}|\Jl\cap J_B^\ell\cap C|\cdot x_{C}^{*}\\
        &\leq  \sum_{C\in\calC(B)}4/\eps^4 \cdot x_{C}^{*}\leq 4/\eps^4.
    \end{align*}
    A union over the $\log_2 T +1$ values of $\ell$ shows $\EE[|\Jt(B)|]\leq 4/\eps^4 \cdot (\log_2 T +1)$
\end{proof}

Consider next the remaining \emph{short} jobs $\Jshort(B):=J(B)\setminus \Jt(B)$. For these jobs, intuitively, we have useful concentration properties. Let us initially set $Z_{j,B}=1$ for all $j\in \Jshort(B)$. For each such $j$ we define one or two (unlikely) bad events: when any one of those events happens, we set $Z_{j,B}=0$. The first bad event $\calE_{total}$ happens if the total processing time of the short jobs assigned to
$B$ is larger than $|B|$, i.e., 
$$
\sum_{j'\in \Jshort(B)}p_{\ale{j'}}=\sum_{j'\in \Jl: p_{j'}\leq \eps^4|\tw(j')\cap B|}p_{j'}\cdot Y_{j',B}>|B|.
$$
When $\calE_{total}$ happens, we set $Z_{j,B}=0$ for all $j\in \Jshort(B)$. The event $\calE_{total}$ is the unique bad event for the considered jobs that span $B$, i.e., such that $B\subseteq tw(j)$.


Assume now that some $j\in \Jshort(B)$ does not span $B$ but that $\tw(j)$ intersects
$B$ from the right, i.e., $s<r_{j}<t$. Then we define the additional
bad event $\calE_{right}(j)$ that the total processing time of the jobs $j'\in \Jshort(B)$ with $r_{j}\le r_{j'}$ is larger than $|\tw(j)\cap B|$. The intuition is
that $j$ and all these jobs need to be processed during $\tw(j)\cap B$.
Formally, $\calE_{right}(j)$ happens if 
$$
\sum_{j'\in \Jshort(B): r_j\leq r_{j'}}p_{j'}=
\sum_{j'\in\Jl:p_{j'}\leq \eps^4|\tw(j')\cap B|, r_j\leq  r_{j'}}p_{j'}\cdot Y_{j,B}>|\tw(j)\cap B|.
$$
If $\calE_{right}(j)$ happens, we set $Z_{j,B}:=0$. 
We define an analogous second bad event $\calE_{left}(j)$ if $j$ does not span $B$
but $\tw(j)$ intersects $B$ from the left, i.e., $s<d_{j}<t$, and
we set $Z_{j,B}$ accordingly. %

We can show that for each job $j\in \Jshort(B)$ the defined
bad events are very unlikely and, hence, if $Y_{j,B}=1$ then most
likely also $Z_{j,B}=1$.
\begin{lem}\label{lem:pseudoConcentration}
For each block $B\in\calB$ and \aw{each} job $j\in \Jl$ with $p_j\leq \eps^{\ale{4}}|\tw(j)\cap B|$ and $\Pr[Y_{j, B}=1]>0$, it
holds that $\Pr[Z_{j,B}=1|Y_{j,B}=1]\ge 1-2\eps.$
\end{lem}

\begin{proof}
    Let $B=[s,t)\in \calB$ and $j\in \Jl$ such that $p_j\leq \eps^4(t-s)$. 
    We condition on the event that $Y_{j, B}=1$.

    First we bound the probability of the event \ale{ $\calE_{total}$, i.e.}
    $X:=\sum_{j'\in \Jshort(B)}p_{j'}>|B|$,
    using Chernoff's bound. 
    For any configuration $C\in \calC(B)$ we have $\sum_{j'\in C}p_{j'}\leq t-s$. This implies the following:
    \begin{align*}\label{eq:expectationpreservance}
        \EE[X | Y_{j,B}=1]
        &=p_j+\sum_{j'\in \Jl\setminus \{j\}: p_{j'}\leq \eps^4|\tw(j)\cap B|}p_{j'}\cdot \EE[Y_{j', B}]\\
        &\leq p_j +\sum_{j'\in \Jl\setminus \{j\}: p_{j'}\leq \eps^4|\tw(j)\cap B|} p_{j'}\cdot (1-2\eps)\sum_{C\in\calC(B):j'\in \ale{J_C}}x_{C}^{*}\\
        &\leq  \eps^4(t-s)+(1-2\eps)\sum_{C\in\calC(B)}\sum_{j'\in \ale{J_C}\cap \Jl : p_{j'}\leq \eps^4|\tw(j)\cap B|}p_{j'}\cdot x_{C}^{*}\\
        &\leq  \eps(t-s)+(1-2\eps)\sum_{C\in\calC(B)}(t-s)\cdot x_{C}^{*}\\
        &=(1-\eps)(t-s).
    \end{align*}
    For $j'\in \Jl$ such that $\ale{p_{j'}}\leq \eps^4(t-s)$, let $Y_{j', B}'=p_{j'}\cdot Y_{j', B}$. 
    Then the random variables $Y_{j', B'}$ for the $j'$ of the above type are independent, their sum $X$ is bounded by $(1-\eps)(t-s)$ and each random variable is bounded by $\eps^4(t-s)$. 
    Therefore, we can apply Chernoff's bound and obtain:
    \begin{align*}
        \Pr[X>(t-s)]&\leq \exp\left(-\frac{\eps^2\cdot (1-\eps)(t-s)}{3 \cdot \eps^4(t-s)}\right)\\
        &\leq \exp\left(-\frac{1-\eps}{3\cdot \eps^2}\right)\\
        &\leq \exp\left(-1/\eps\right)\leq \eps.
    \end{align*}
    We used $\eps\leq 1/4$ in the third inequality.\alnote{$\eps\leq 1/4$ needed here}
    This already implies the lemma for all jobs spanning $B$, as the event \ale{$\calE_{total}$} is the only bad event for such jobs.
    
    Now, suppose that $s<r_j<t$: a symmetric argument holds when $s<d_j<t$.
    We need to upper bound the probability of the event \ale{$\calE_{right}$, i.e.,}
    $$
    X':=\sum_{j'\in \Jshort(B):r_{j}\le r_{j'}}p_{j'}>|\tw(j)\cap B|=t-r_j.
    $$
    For any configuration $C\in \cal C(B)$, we have
    \begin{equation*}
        \sum_{j'\in \ale{J_C}\cap \Jl: p_{j'}\leq \eps^4|\tw(j)\cap B|, r_{j}\le r_{j'}}p_{j'}\leq (t-r_j)
    \end{equation*}
    as all these jobs have to be scheduled during $[r_j, t)$.
    A calculation analogous to the one for $X$ shows that $\EE[X']\leq (1-\eps)\cdot (t-r_j)$. Note that each job $j'\in \Jshort(B)$ with $r_{j}\le r_{j'}$ fulfills $p_{j'}\leq \eps^4 \cdot |\tw(j')\cap B|\leq \eps^4(t-r_j)$ as the jobs are short and $\tw(j')\cap B\subseteq \tw(j)\cap B$.
    As for $X$ we use Chernoff's bound, but this time applied to the random variables $Y_{j', B}'$ only for $j'\in \Jl$ with $p_{j'}\leq \eps^4|\tw(j)\cap B|$ and $r_{j}\le r_{j'}$.
    We obtain $\Pr[X'>t-r_j]\leq \eps$.
    Together, we obtain
    \begin{equation*}
        \Pr[Z_{j, B}=1|Y_{j, B}=1]\geq 1-\Pr[X>t-s]-\Pr[X'>t-r_j]\geq 1-2\eps.
    \end{equation*}
    This completes the proof.
\end{proof}

Let $\APX(B):=\{j\in\Jl:Z_{j,B}=1\}\subseteq J(B)$ be all the jobs that we finally
assign to $B$. We show that we can compute a feasible schedule
for them in $B$ (with a simple greedy algorithm).
\begin{lem}\label{lem:pseudoGreedy}
For each block $B\in\calB$ we can compute a schedule for $\APX(B)$
in $B$ in \fab{polynomial time}.
\end{lem}

\begin{proof}
Consider a block $B=[s,t)$.
Let us partition the jobs $\APX(B)$ into the ones $\APX^{\text{left}}(B):=\{j\in \APX(B):s<d_j<t\}$ whose time window intersects $B$ from the left, the ones $\APX^{\text{right}}(B):=\{j\in \APX(B):s<r_j<t\}$ whose time window intersects $B$ from the right, and the remaining ones $\APX^{\text{mid}}(B):=\APX(B)\setminus (\APX^{\text{left}}(B) \cup \APX^{\text{right}}(B))$, i.e., the ones whose time window spans $B$.

First, we order the jobs $\APX^{\text{left}}(B)$ non-decreasingly by their deadline. We schedule them in this order without idle time between them, such that the first job is starting at $s$. 
Next, we order the jobs $\APX^{\text{right}}(B)$ non-decreasingly by their release time. 
We also schedule them in this order without idle time between them, but such that the last job ends at time $t$.
Finally, we schedule the jobs $\APX^{\text{mid}}(B)$ in arbitrary order between the jobs $\APX^{\text{left}}(B)$ and the jobs $\APX^{\text{right}}(B)$.

Clearly, this computation can be done in polynomial time. 
It remains to show that this yields a valid schedule.
If $\APX(B)=\emptyset$, then there is nothing to show.
So suppose that this is not the case. It is sufficient to prove the following two claims. First, we claim that there is enough space between $\APX^{\text{left}}(B)$ and $\APX^{\text{right}}(B)$ to schedule $\APX^{\text{mid}}(B)$. 
Second, we claim that the jobs in $\APX^{\text{left}}(B)\cup \APX^{\text{right}}(B)$ are scheduled within their time window.

We start with the first claim. 
As $\APX(B)\neq \emptyset$, there exists $j\in \APX(B)$, which therefore fulfills $Z_{j, B}=1$. \ale{Thus the event $\calE_{total}$ does not occur.}
This implies
\begin{equation*}
    \sum_{j\in\APX(B)}p_{j}\leq \sum_{j\in \Jshort(B)}p_{j}\leq |B|,
\end{equation*}
so there is enough space to schedule the jobs $\APX^{\text{mid}}$ between $\APX^{\text{left}}$ and $\APX^{\text{right}}$, given that the latter two sets are scheduled without any idle time in the leftmost and rightmost portion of $B$, respectively.

For the second claim, consider any $j\in \APX^{\text{right}}$ (a symmetric argument works for any $j\in \APX^{\text{left}}$). 
We have $Z_{j, B}=1$ and \ale{thus the event $\calE_{right}(j)$ does not occur. Therefore we obtain}
\begin{equation*}
    \sum_{j'\in \APX(B):r_{j}\le r_{j'}}p_{j'}\leq \sum_{j'\in \Jshort(B):r_{j}\le r_{j'}}p_{j'}\leq t-r_j.
\end{equation*}
Note that all jobs $j'$ scheduled after (and including) $j$ fulfill $\ale{r_{j}\leq r_{j'}}$. Thus their total processing time is at most $t-r_j$.
As the last job is completed at time $t$ and there is no idle time between the jobs in $\APX^{\text{right}}$, this implies that $j$ is started not earlier than $t-(t-r_j)=r_j$. 
This shows that $j$ is scheduled within its time window.
\end{proof}

We apply the procedure above to each block $B\in\calB$. Our solution
is the set of jobs $\APX:=\bigcup_{B\in\calB}APX(B)$ together with the associated schedule as described above. Combining the above results yields the following lemma.
\begin{lem}
\label{lem:rounding-good} 
In \fab{polynomial time} we can compute a feasible solution to
the given instance whose expected number of jobs is at least 
\[
\Big((1-2\eps)\sum_{j\in\Jl}y_{j}^{*}\Big)-4(\log_2 T+1)/\eps^4\cdot |\calB|\ge(1-O(\varepsilon))\sum_{j\in\Jl}y_{j}^{*}-O(\eps)\sum_{j\in J}y_{j}^{*}.
\]
\end{lem}

\begin{proof}
    By Lemma~\ref{lem:pseudoGreedy} we compute a feasible schedule in time $n^{O(1)}$ for each block $B\in \calB$. 
    As each job is assigned only to one block, every job is scheduled at most once and thus this yields a feasible schedule altogether. 
    It remains to show that the computed schedule has a high enough number of jobs.
    {\allowdisplaybreaks
    \begin{align*}
        \EE\left[\sum_{B\in \calB}|\APX(B)|\right]&=\EE\left[\sum_{B\in \calB} \sum_{j\in \Jl}Z_{j, B}\right]\\
        &=\sum_{B\in \calB} \sum_{j\in \Jl}\Pr[Z_{j, B}=1|Y_{j, B}=1]\cdot \Pr[Y_{j, B}=1]\\
        &\overset{\text{Lem. }\ref{lem:pseudoConcentration}}{\geq}\sum_{B\in \calB} \sum_{j\in \Jl: p_j\leq \eps^4|\tw(j)\cap B|}(1-2\eps)\cdot \Pr[Y_{j, B}=1]\\
        &\geq(1-2\eps)\sum_{B\in \calB}\left( \sum_{j\in \Jl}\EE[Y_{j, B}]-\EE[|\Jt(B)|]\right)\\
         &\overset{\text{Lem. }\ref{lem:pseudoRemovalLongJobs}}{\geq} (1-2\eps)\sum_{B\in \calB}\left( \sum_{j\in \Jl}\EE[Y_{j, B}]-4(\log_2 T+1)/\eps^4\right)\\
         &\geq(1-2\eps)\sum_{j\in \Jl}\sum_{B\in \calB}\EE[Y_{j, B}]-\ale{\Delta}\cdot |\calB|\\
         &=(1-2\eps)\sum_{j\in \Jl}(1-2\eps)y_j^*-\ale{\Delta}\cdot |\calB|\\
        &\overset{\text{Lem. }\ref{lem:good-block-superblock-partition}}{\geq}(1-4\eps)\sum_{j\in \Jl}y_j^*-\eps \cdot |\OPT'|\\
        &\overset{\text{Lem. }\ref{lem:solveLP}}{\geq}(1-4\eps)\sum_{j\in \Jl}y_j^*-2\eps \sum_{j\in J}y_j^*.
    \end{align*}
    }
    This completes the proof.
\end{proof}

Finally, we output the best of the two solutions computed with the algorithm of this section and with the algorithm from Section~\ref{sec:firstRounding} (using the same optimal solution to the configuration LP for $\Delta=4(\log_2T+1)/\eps^4$ in both cases). Then, Lemmas~\ref{lem:matchingExistence},
\ref{lem:solve-large-LP}, and \ref{lem:rounding-good} yield 
Theorem~\ref{thr:mainPseudopoly}.

\begin{proof}[Proof of Theorem~\ref{thr:mainPseudopoly}]
    By Lemma~\ref{lem:solve-large-LP}, we solve the LP for $\Delta=4(\log_2 T+1)/\eps^4$ in time $(nT)^{O_{\eps}(1)}$.
    By Lemma~\ref{lem:solve-large-LP}, the computed solution has profit at least $(1-2\eps)|\OPT|$.
    If $\sum_{j\in \Js}y_j^*\geq \frac{|\OPT|}{5}$,
    we know from Lemma~\ref{lem:matchingExistence}, that the rounding procedure from Section~\ref{sec:firstRounding} yields a schedule for at least
    \begin{align*}
        (1-O(\eps))\left(\frac{3}{4}\sum_{j\in J}y_j^*+\frac{1}{4}\sum_{j\in \Js}y_j^*\right)
        &\geq (1-O(\eps))\left(\frac{3}{4}\cdot |\OPT|+\frac{1}{20} \cdot |\OPT|\right) \\
        &\geq (1-O(\eps))\cdot \frac{4}{5} \cdot |\OPT|
    \end{align*}
    jobs.
    On the other hand, if $\sum_{j\in \Js}y_j^*< \frac{|\OPT|}{5}$, we have
    \begin{equation*}
        \sum_{j\in \Jl}y_{j}^*\geq (1-O(\eps))\cdot \frac{4}{5}\cdot |\OPT|.
    \end{equation*}
    This implies by Lemma~\ref{lem:rounding-good}, that the second rounding procedure from Section~\ref{sec:secondRounding} yields a schedule for at least $(1-O(\eps))\cdot \frac{4}{5} \cdot |\OPT|$ jobs.
    In both cases, we obtain a $(5/4+O(\eps))$-approximate solution. The claim follows by rescaling $\eps$ appropriately.
\end{proof}

\section{Extension to multiple machines}
\label{sec:multiple}

In this section, we describe how we extend our polynomial time $(4/3+\eps)$-approximation
and our pseudo-polynomial time $(5/4+\eps)$-approximation algorithms
to the setting of multiple \fab{(identical)}\fabr{Once within parentheses is enough} machines.

First, we define the problem formally on multiple machines. In the
input, we are given a set of jobs $J$ with given processing times,
release times, and deadlines as in the case of one machine. Additionally, \fab{we are given} a value $m\in\NN$ that denotes the given number of (identical)
machines. Our goal is again to compute a subset of jobs $J'\subseteq J$;
however, now we also need to compute a partition $J'=J'_{1}\dot{\cup}...\dot{\cup}J'_{m}$
where for each $i\in[m]$ the set $J'_{i}$ corresponds to the jobs
we assign to machine $i$. Like before, for each job $j\in J'$ we
need to compute a start time $s(j)\in\NN$ such that $[s(j),s(j)+p_{j})\subseteq\tw(j)$.
We require that for each machine $i\in[m]$ and for any two jobs $j,j'\in J'_{i}$
that $[s(j),s(j)+p_{j})\cap[s(j'),s(j')+p_{j'})=\emptyset$. Note
that we do not require this for two jobs assigned to different machines
$i,i'$. As before, the objective is to maximize $|J'|$. We call
the resulting problem Throughput Maximization on $m$ machines.

It was shown in \cite{im2020breaking} that for each $\eps>0$ there
is a $(1-O(\sqrt{(\log m)/m})-\eps)^{-1}$-approximation algorithm.
\begin{thm}[\cite{im2020breaking}]
\label{thm:small-m}For any $\varepsilon>0$, there exists a polynomial
time $(1-O(\sqrt{(\log m)/m})-\eps)^{-1}$approximation for Throughput
Maximization on $m$ machines.
\end{thm}

Hence, for each $\eps>0$ there is a constant $m_{\eps}\in\NN$, \fab{$m=O_\eps(1)$,} such
that the theorem above yields a $(1+O(\eps))$-approximation algorithm
if $m>m_{\eps}$. Therefore, in the following we assume that we are
given a constant $\eps>0$ and an instance of Throughput Maximization
on multiple machines in which $m\le m_{\eps}$. For this setting,
we give a polynomial time $(4/3+\eps)$-approximation and a pseudo-polynomial
time $(5/4+\eps)$-approximation algorithm.

\subsection{Construction of blocks and superblocks}

First, we generalize our construction of the blocks and superblocks
to the setting of multiple machines. We start with a corresponding
statement of Lemma~\ref{lem:block-preprocessing} that constructs
the elementary blocks on multiple machines. We ensure that there is
a near-optimal solution which schedules at most $O_{\eps}(m)$ jobs
during each block. Note that since $m\le m_{\eps}=O_{\eps}(1)$ this
is bounded by a constant.
\begin{lem}
\label{lem:block-preprocessing-m} Assume that $|\OPT|\geq (\frac{6m}{\eps})^{3}$.
In polynomial time we can compute a partition of $[0,T)$ into a set
of blocks $\tilde{\calB}_{0}$ such that there exists a \fab{set of jobs} $\widetilde{\OPT}''\subseteq J$ \fab{and a feasible schedule of them} with the following properties:
\begin{enumerate}
[label=(A\arabic*)]
\item \label{item:cor-opt-m} $|\widetilde{\OPT}''|\geq(1-\eps)|\OPT|$,
\item \label{item:cor-restricted-m} each job $j\in\widetilde{\OPT}''$
is scheduled in some block $B\in\tilde{\calB}_{0}$,
\item \label{item:cor-block-bound-m} for each block $B\in\tilde{\calB}_{0}$
there are at most \ale{$\tilde{K}_{0}=(m/\eps)^{O(m/\eps\log(m/\eps))}$}\alnote{I changed the value of $\tilde{K}_0$}
jobs of $\widetilde{\OPT}''$ that are scheduled in $B$,
\item \label{item:cor-bound-opt-m} $|\widetilde{\OPT}''|\ge m\cdot|\tilde{\calB}_{0}|/\eps$.
\end{enumerate}
\end{lem}
\begin{proof}    \alnote{This is the only really new proof.}
    Again, we use the result by \cite{chuzhoy2006approximation} as described in \ant{Appendix}~\ref{sec:appendix-block}. 
    Their result holds for a more general problem than Maximum Throughput, called the Job Interval Selection Problem (JISP). 
    Instead of a release time, a deadline and a processing time, each job $j\in J$ in JISP has a set $\calI_j$ of intervals and it can be scheduled for at most one interval in $\calI_j$ instead of all intervals of length $p_j$ contained in $[r_j, d_j)$.
    So Maximum Throughput on one machine is the special case of JISP where $\calI_j=\{[t, t+p_j):[t, t+p_j)\subseteq [r_j, d_j)\}$.
    Intuitively, in the multiple machines case we have one separate interval of length $T$ for each machine, i.e., a job scheduled on the $i$-th machine is scheduled during the interval $[(i-1)T, iT)$. Formally, for each $j\in J$ and each $i\in [m]$ let
    $$\calI_{j, i}:=\{[t, t+p_j): t\in \NN \wedge [t, t+p_j)\subseteq [r_j+(i-1)T, d_j+(i-1)T)\}$$
    Also let $\calI_j:=\bigcup_{i=1}^m \calI_{j, i}$ denote the intervals in which $j$ can be scheduled.
    A schedule for this JISP instance corresponds to an instance of Maximum Throughput \fab{on $m$ machines}:  A job scheduled during the interval $[(i-1)T+t, (i-1)T+t+p_j)\in \calI_{j, i}$ in the JISP instance is scheduled on machine $i$ during the interval $[t, t+p_j)$ in the \fab{Maximum} Throughput instance.
    
    Now we apply the algorithm from \cite{chuzhoy2006approximation}, for $\eps'=\eps/(2m)$ instead of $\eps$. Note that even though $\calI_j$ is only pseudopolynomial in the input size, the described greedy algorithm can still be implemented in polynomial time (in the size of the given input instance), as each job only has $m$ different time windows.
    We do the same steps as in the proof of Lemma~\ref{lem:block-preprocessing}, which yields, assuming that $|\OPT|\geq (\frac{6m}{\eps})^{3}$,  a partition of $[0,mT)$ into a set of blocks $\hat{\calB}_{0}$ such that there exists a set of jobs $\widehat{\OPT}\subseteq J$ and a feasible schedule of them with the following properties:
\begin{enumerate}[label=(C\arabic*)]
    \item\label{item:cor-opt-m-proof} $|\widehat{\OPT}|\geq(1-\eps/(2m))|\OPT|$, 
    \item \label{item:cor-restricted-m-proof} each job $j\in\widehat{\OPT}$ is scheduled in some block $B\in\hat{\calB}_{0}$,
    \item\label{item:cor-block-bound-m-proof} for each block $B\in\hat{\calB}_{0}$ there are at most $\tilde{K}_{0}'=(m/\eps)^{O(m/\eps\log(m/\eps))}$ jobs of $\widehat{\OPT}$ that are scheduled in $B$, 
    \item\label{item:cor-bound-opt-m-proof} $|\widehat{\OPT}|\ge 2m\cdot |\hat{\calB}_{0}|/\eps$.
\end{enumerate}
    Now we can define the blocks $\tilde{B}_0$ for the given instance of Maximum Throughput on $m$ machines. \fab{Intuitively, we cut the blocks $\hat{\calB}_0$ at multiples of $T$, and then we shift the blocks corresponding to the interval $[(i-1)T,iT)$ into the interval $[0,T)$ so as to obtain a partition into blocks for machine $i$. However, we need a consistent block partitioning over the different machines. 
    \ale{Therefore, whenever a blocks starts or ends for some machine $i$, we use this a start or end point of a new block for each machine.}
    Formally,} let $P=\{s, t:B=[s, t)\in \hat{\calB}_0\}$ denote all start and endpoints of blocks and let $P_{mod}=\{t'\in [0, T]:\exists t\in P \text{ such that } t\equiv \fab{t'} \text{ modulo } T\}$ denote the corresponding points modulo $T$.
    Assume that $P_{mod}=\{t_0=0, \dots, t_r=T\}$ such that for all $k\in [r]$ we have $t_{k-1}<t_{k}$.
    We define by $\tilde{\calB}_0=\{[t_{k-1}, t_k):k\in [r]\}$ the partition of $[0, T)$ into blocks.
    By construction we have \ale{$|P_{mod}|\leq |P|$ and thus} $|\tilde{\calB}_0|\ale{\leq} |\hat{\calB}_0|$.
    We define a schedule $\widetilde{\OPT}'$ for the \ale{Maximum} Throughput instance obtained from $\widehat{\OPT}$ as follows. Consider a job $j \in \widehat{\OPT}$ and assume that it is scheduled during some interval
    $[(i-1)T+t, (i-1)T+t+p_j)\in \calI_{j, i}$ for some value $i$. We include $j$ in our set $\widetilde{\OPT}'$ and we schedule $j$ on machine $i$
    during the interval $[t, t+p_j)$. We apply this reasoning to each job $j\in \widehat{\OPT}$.
    Clearly, $\widetilde{\OPT}'$ is feasible and we have $|\widetilde{\OPT}'|=|\widehat{\OPT}|$.
    Furthermore, let $\widetilde{\OPT}''$ denote the schedule obtained from $\widetilde{\OPT}'$ \fab{by} removing all jobs that are not completely scheduled within one block from $\tilde{\calB}_0$.
    
    It remains to show that this schedule fulfills the properties stated in the lemma.
    Property \ref{item:cor-restricted-m} is fulfilled by construction of $\widetilde{\OPT}''$.
    Let $k\in [r-1]$. Then there exists $i\in [m]$ such that a block in  $\hat{\calB}_0$ begins at $(i-1)T+t_k$. Thus there are at most $m-1$ jobs $j\in \widetilde{\OPT}'$, which are scheduled partially in both the blocks $[t_{k-1}, t_k)$ and $[t_{k}, t_{k+1})$.
    Thus we have $|\widetilde{\OPT}''|\geq |\widetilde{\OPT}'|-(m-1)|\tilde{\calB}_0|$.
    Using \ref{item:cor-opt-m-proof} and \ref{item:cor-bound-opt-m-proof}, we obtain 
    \begin{align*}
        |\widetilde{\OPT}''|&\geq |\widetilde{\OPT}'|-(m-1)|\tilde{\calB}_0|
        \geq |\widehat{\OPT}|-(m-1)|\hat{\calB}_0|\\
        &\ale{\geq (1-\eps/(2m))|\OPT|-(m-1)\frac{\eps |\OPT|}{2m}}\\
        &\geq (1-\eps)|\OPT|
    \end{align*}
    \ale{This} yields \ref{item:cor-opt-m}.
    
    By construction $\widehat{\OPT}$ schedules at most $\tilde{K}_0'$ jobs during each block in $\hat{\calB}_0$. Thus $\widetilde{\OPT}''$ schedules at most $\tilde{K}_0:=m\cdot \tilde{K}_0'$ jobs during each block in $\tilde{\calB}_0$, yielding \ref{item:cor-block-bound-m}.
    Finally, \ale{using \ref{item:cor-opt-m}} we have that  $m\cdot |\tilde{\calB}_0|/\eps \fab{\leq} m\cdot |\hat{\calB}_0|/\eps\leq \frac{1}{2}|\widehat{\OPT}|\ale{\leq \frac{1}{2}|\OPT| }\leq |\widetilde{\OPT}''|$ which yields \ref{item:cor-bound-opt-m} and, hence, completes the proof.

\end{proof}
We use the same definition of a block-superblock partition as in the
case of one machine only. Based on our elementary blocks $\tilde{\calB}_{0}$
from Lemma~\ref{lem:block-preprocessing-m}, we compute $1/\eps$
block-superblock partitions in exactly the same way as in Lemma~\ref{lem:good-block-superblock-partition}.
As we have $|\widetilde{\OPT}''|\ge m\cdot|\tilde{\calB}_{0}|/\eps$  by Lemma~\ref{lem:block-preprocessing-m} instead of $|\OPT''|\ge |\calB_0|/\eps$ shown in Lemma~\ref{lem:block-preprocessing}, we also get the stronger bound $|\widetilde{\OPT}'|\ge m\frac{\Delta}{\eps}|\tilde{\calB}|$ instead of the bound $|\OPT'|\ge \frac{\Delta}{\eps}|\calB|$ in Lemma~\ref{lem:good-block-superblock-partition}. For the same reason we also get $|\widetilde{\OPT}'|\ge m\frac{\tilde{K}}{\eps^{6}}|\tilde{\calS}|$ instead of $|\OPT'|\ge \frac{K}{\eps^6}|\calS|$. However, the bound $|\widetilde{\OPT}'|\ge \frac{\tilde{K}}{\eps^{6}}|\tilde{\calS}|$ 
(as in the case of a single machine)
is
sufficient for our analysis later.
\begin{lem}
\label{lem:good-block-superblock-partition-m}Let $\Delta\in\NN$
be a given parameter, and assume that $|\OPT|\geq 5\tilde{K}_0\Delta(2\tilde{K}_{0}/\eps^{5})^{1/\eps}$.
In time $\Delta^{O(1)}n^{\fab{O_{\eps}(1)}}$ we can compute at most
$1/\eps$ block-superblock partitions and for each one of them a value
$\tilde{K}$ with $\tilde{K}\le \Delta(2\tilde{K}_{0}/\eps^5)^{1/\eps}$
such that for at least one computed partition $(\tilde{\calB},\tilde{\calS})$
there exists a set $\widetilde{\OPT}'\subseteq J$ and a corresponding
feasible schedule with the following properties:
\begin{enumerate}
[label=(B\arabic*)]
\item \label{item:block-superblock-partition-lb-m} $|\widetilde{\OPT}'|\geq(1-2\varepsilon)|\OPT|$,
\item \label{item:border-or-spanning-m} each job $j\in\widetilde{\OPT}'$
is scheduled in $B_{j,L}$, or in $B_{j,R}$, or in a block $B\in\tilde{\calB}$
for which $j$ spans the superblock $S\in\tilde{\calS}$ containing
$B$,
\item \label{item:configuration-size-m} for each block $B\in\tilde{\calB}$
there are at most $\tilde{K}$ jobs from $\widetilde{\OPT}'$ that
are scheduled in $B$,
\item \label{item:fewBlocks-m} $|\widetilde{\OPT}'|\ge m\frac{\Delta}{\eps}|\tilde{\calB}|$
and $|\widetilde{\OPT}'|\ge \frac{\tilde{K}}{\eps^{6}}|\tilde{\calS}|$.
\end{enumerate}
\end{lem}

Like before, we guess the block-superblock partition from Lemma~\ref{lem:good-block-superblock-partition-m}
satisfying properties \ref{item:block-superblock-partition-lb-m}-\ref{item:fewBlocks-m}
and the corresponding value of $\tilde{K}$. Let $(\tilde{\calB},\tilde{\calS})$
and $\tilde{K}$ denote these values in the following.

\subsection{The linear program}

Based on $(\tilde{\calB},\tilde{\calS})$ and $\tilde{K}$ we define
the configuration-LP, similarly as in the case of one machine. Formally,
a configuration $C$ is specified by a pair $(B_{C},J_{C})$, where
$B_{C}=[s,t)\in\tilde{\calB}$ is a block and $J_{C}$ is a subset
of jobs that can be feasibly scheduled inside $B_{C}$ on $m$ machines
according to Lemma~\ref{lem:good-block-superblock-partition}; formally,
we require that $|J_{C}|\leq\tilde{K}$ and that, for each $j\in J_{C}$,
the block $B_{C}$ is either a boundary block for $j$ or $B_{C}$
is contained in a superblock $S$ spanned by $j$ (i.e., $S\subseteq\tw(j)$).
Moreover, \aw{we fix for $C$ a}
\fabr{More than assuming, we fix one such partition (like we fix the schedule). We see it as associated implicitly with the schedule. The LP solver will find it}\awr{rephrased it}
partition $J_{C}=J_{C,1}\dot{\cup}...\dot{\cup}J_{C,m}$
of the jobs in $J_{C}$ for the $m$ machines and a function $s_{C}:J_{C}\rightarrow\{s,\ldots,t-1\}$
denoting the jobs' starting times such that the mentioned partition
of $J_{C}$ and the function $s_{C}$ together yield a feasible schedule
of $J_{C}$ in $B_{C}$. Similarly as before, for each block $B$
we define the set $\tilde{\calC}(B)$ to be the set of all configurations
for $B$ and we set $\tilde{\calC}\coloneqq\bigcup_{B\in\tilde{\calB}}\fab{\tilde{\calC}}(B)$.
In our LP, for each configuration $C\in\tilde{\calC}$ we introduce
a variable $x_{C}$ representing whether we select $C$ for its corresponding
block. We introduce constraints to model that we select each job at
most once and we select one configuration for each block.
\begin{equation}
\tag{\textsf{LP}}\begin{split}\max\sum_{C\in\tilde{\calC}}|J_{C}|\cdot x_{C}\\
\text{s.t.}\sum_{C\in\tilde{\calC}:j\in J_{C}}x_{C} & \leq1\quad\text{\ensuremath{\forall}}j\in J\\
\sum_{C\in\tilde{\calC}(B)}x_{C} & =1\quad\text{\ensuremath{\forall}}B\in\tilde{\calB}\\
x_{C} & \geq0\quad\text{\ensuremath{\forall}}C\in\tilde{\calC}.
\end{split}
\label{eq:configuration-lp-m}
\end{equation}

Similarly as before, we can show that the optimal objective function
value of (\ref{eq:configuration-lp-m}) is almost as large as $|\OPT|$.

\begin{lem}
\label{lem:valueLP-m} Let $(\tilde{\calB},\tilde{\calS})$ be a block-superblock
partition satisfying properties \ref{item:block-superblock-partition-lb-m}-\ref{item:fewBlocks-m}
of Lemma~\ref{lem:good-block-superblock-partition-m}. Then the optimal
objective function value of the associated configuration LP is at
least %
\mbox{%
$(1-2\eps)|\OPT|$%
}.
\end{lem}

Also, we can solve (\ref{eq:configuration-lp-m}) sufficiently fast
by generalizing the algorithms due to Lemmas~\ref{lem:solveLP} and
\ref{lem:solve-large-LP}.
\begin{lem}
\label{lem:solveLP-m} Let $\tilde{K}$ denote the maximum number
of jobs in a configuration. Then in time 
$\aw{(2m)}^{O(\tilde{K})}\cdot(nT)^{O(1)}$\fabr{I'm not sure about the multiplicative factor $m$ in the exponent in the 2 pseudopoly results. I think $\tilde{K}^m$ and $|\OPT|^m$ look more logical}\awr{I think the term $\aw{m}^{O(\tilde{K})}$ is good since there are at most $m^{\tilde{K}}$ options how the $\tilde{K}$ colors are distributed over the $m$ machines} we can compute an optimal
solution $(x_{C}^{*})_{C\in\calC}$ to the configuration LP together
with a feasible schedule $s_{C}$ for each configuration $C\in\calC$
with $x_{C}^{*}>0$. 
\aw{If $m=O_\eps(1)$ and $\tilde K = O_{\eps,m}(1)$ and  we can do this also in polynomial time.}\awr{put this to address issue with polytime etc.}
Also, Maximum Throughput on $m$ machines can
be solved exactly in time $\aw{(2m)}^{O(|\OPT|)}(nT)^{O(1)}$.
\end{lem}
\begin{proof}
    \aw{We can enumerate all configurations in time $n^{O(\tilde{K})}$, similarly as in Lemma~\ref{lem:solveLP}. For each block and each subset of at most $\tilde{K}$ jobs, we consider each partition of the jobs on the $m$ machines, and for each machine $i\in [m]$, we compute a schedule for its jobs like in the case of one machine (if a schedule exists). 
    Given all configurations, we can solve the LP in polynomial time.}

    \aw{Next, we show that we can solve the LP also in time $\aw{(2m)}^{O(\tilde{K})}\cdot(nT)^{O(1)}$. We proceed as in the proof of Lemma~\ref{lem:solve-large-LP}. Recall that there we need to compute a configuration of maximal weight for a block via the recursion~\eqref{eq:COloringDPrecursion}. We cannot directly apply this in our setting. Instead, we guess how the $\tilde{K}$ colors are distributed among the $m$ machines. For this there are at most $m^{\tilde{K}}$ options. Then, for each machine we use the DP specified by \eqref{eq:COloringDPrecursion} to compute a maximum weight set of jobs that can be scheduled on this machine and that have the respective colors. This way, we can also compute a maximal weight configuration for a block. The remaining proof is the same as for Lemma~\ref{lem:solve-large-LP}.}

    \aw{Finally, we can solve Maximum Throughput \fab{on} $m$ machines exactly in time $(2m)^{O(|\OPT|)}(nT)^{O(1)}$ using the same idea.}
\end{proof}

\subsection{Polynomial time approximation algorithm on multiple machines}\label{subsec:poly-m-machines}

Our polynomial time $(4/3+\eps)$-approximation algorithm is almost
identical to the algorithm for the case of only one machine. We choose
as before $\Delta:=1$ and solve the configuration-LP~(\ref{eq:configuration-lp-m}). \aw{If $m=O_{\eps}(1)$ we can do this in polynomial time by Lemma~\ref{lem:solveLP-m}}
since $\tilde{K}=O_{\eps\ale{, m}}(\ale{1})$
\alnote{I adjusted this to the current value of $\tilde{K}_0$}. Let $(x_{C}^{*})_{C\in\tilde{\calC}(B)}$
denote an optimal solution for it. Then, for each block $B\in\tilde{\calB}$
independently, we sample one configuration $C^{*}(B)\in\tilde{\calC}(B)$
with respect to the distribution given by the vector $(x_{C}^{*})_{C\in\tilde{\calC}(B)}$.
Then, we define a set of slots $\tilde{Q}$ and a bipartite graph
analogously to the case of one machine. We compute a maximum matching
$\tilde{M}^{*}$ in this bipartite graph and output the jobs that
are matched in $\tilde{M}^{*}$.

To show that $\tilde{M}^{*}$ is sufficiently large, we prove a lemma
that generalizes Lemma~\ref{lem:solve-large-LP}. The sets $\Js$
and $\Jl$ are defined like in the case of one machine and for each
job $j\in J$ we define $y_{j}^{*}:=\sum_{C\in\tilde{\calC}:j\in J_{C}}x_{C}^{*}$
\begin{lem}
\label{lem:matchingExistence-m} For any $\Delta\geq1$ we have
\[
\EE[|\tilde{M}^{*}|]\ge(1-O(\eps))\sum_{j\in\Js}y_{j}^{*}+(3/4-O(\varepsilon))\sum_{j\in\Jl}y_{j}^{*}\ge(3/4-O(\varepsilon))|\OPT|.
\]
\end{lem}
\aw{The proof is the essentialy same as the proof of Lemma~\ref{lem:matchingExistence}. Let $\tilde{E}$ denote the edges of the constructed bipartite graph. Again, we define a fractional matching $f: \tilde{E} \rightarrow [0,1]$. Like before, if for some job $j$ we sampled a configuration $C$ for one of its boundary blocks $B_{j,L},B_{j,R}$ such that $j\in J_C$, then we assign $j$ integrally to its corresponding slot. The other jobs are assigned (potentially fractionally) to slots in superblocks that they span. Recall that we crucially used that if a job $j$ can be assigned to a slot in a spanning superblock, then it can be assigned to all slots \ale{longer} than $p_j$ in this superblock as well. This is also true in the case of multiple machines. 
Also, we crucially used concentration bounds for the slots of each of the $1/\eps$ groups within each superblock. However, for these arguments, it is irrelevant whether the slots within a superblock are distributed on one or on multiple machines. The important arguments are that each block contains at most $\tilde{K}$ jobs, that the average number of jobs in a superblock is larger than the maximum number of jobs in a block,
how \ale{the number of} 
slots of certain size ranges 
\ale{is distributed, }
and how \ale{the number of} 
jobs of certain size ranges are distributed when we restrict our attention to the jobs that are not assigned to their respective boundary intervals\awr{changed wording in this sentence}. Thus, we can use the same
 probabilistic arguments to show that we have enough jobs and slots of each relevant size range with sufficiently high  probability. For example, 
similar to the single machine case, we know by Lemma~\ref{lem:block-preprocessing-m} that within each block there are at most $\tilde{K}$ jobs scheduled by any configuration and we also have $|\widetilde{\OPT}'|\ge \frac{\tilde{K}}{\eps^{6}}|\tilde{\calS}|$. These are  the same bounds as in Lemma~\ref{lem:good-block-superblock-partition} (with $\tilde{K}$ instead of $K$). Therefore we can show exactly the same concentration arguments as in the proof of Lemma~\ref{lem:matchingExistence}.}
\awr{please check}
This yields our polynomial time $(4/3+\eps)$-approximation algorithm
on multiple machines.
\begin{thm}
\label{thr:mainPoly-m} For any constants $\eps>0$ and $m\in\NN$ \fab{with $m=O_\eps(1)$}\awr{I find $m=O_\eps(1)$ confusing, $m$ and $\eps$ are both assumed to be constant and this is enough, no?},
there is a polynomial-time randomized $(4/3+\eps)$-approximation
algorithm for Throughput Maximization on $m$ machines.
\end{thm}

Due to Theorem~\ref{thm:small-m} we also obtain a $(4/3+\eps)$-approximation
algorithm for an arbitrary number of machines $m$.
\begin{cor}
\label{cor:mainPoly-m} For any constant $\eps>0$, there is a polynomial-time
randomized $(4/3+\eps)$-approximation algorithm for Throughput Maximization
on $m$ machines.
\end{cor}

\aw{Hence, Corollary \ref{cor:mainPoly-m} yields the first claim of Theorem \ref{thr:multipleMachines}.
}

\subsection{Pseudopolynomial time approximation algorithm on multiple machines}

We describe now how we generalize our pseudo-polynomial time $(5/4+\eps)$-approximation
algorithm to the setting of multiple machines. Similar as in the algorithm
for one machine, we choose $\Delta:=4(\log_{2}T+1)/\eps^{4}$. In
time $(nT)^{\ale{O_{\eps, m}(1)}}$ we check whether $|\OPT|<5\ale{\tilde{K}_0}\Delta(2\ale{\tilde{K}}_{0}/\eps^{5})^{1/\eps}$
and solve the given instance optimally if this is the case, via Lemma~\ref{lem:solveLP-m}.
Otherwise, we solve the configuration-LP (\ref{eq:configuration-lp-m})
in time 
$(nT)^{\ale{O_{\eps, m}(1)}}$\alnote{I think its this running time}
using Lemma~\ref{lem:solveLP-m}. As
above, let $(x_{C}^{*})_{C\in\tilde{\calC}(B)}$ denote an optimal
solution for it.

The alternative rounding algorithm is again very similar to the setting
of one machine. We ignore all the (local) jobs $\Js$ and do not schedule
any of them. For each machine $i\in[m]$, each block $B\in\tilde{B}$
and job $j\in\Jl$, we define $y_{j,i,B}^{*}:=\sum_{C\in\tilde{\calC}(B):j\in J_{C,i}}x_{C}^{*}$
\fab{to} be the fractional amount by which $j$ is assigned to machine $i$
in block $B$ in the optimal LP solution. For each job $j\in\Jl$
independently, we discard $j$ with probability $1-(1-2\eps)\sum_{B\in\tilde{\calB}}\sum_{i=1}^{m}y_{j,i,B}^{*}$,
and otherwise assign it to exactly one block $B$ on one machine $i\in[m]$, \aw{such that each combination of a machine $i$ and a block $B \in \tilde{\calB}$ is selected }
with probability
$(1-2\varepsilon)y_{j,i,B}^{*}$. For each block $B\in\tilde{\calB}$
and each machine $i\in[m]$ we define a random variable $Y_{j,i,B}\in\{0,1\}$,
modeling whether we assign $j$ on machine $i$ to block $B$ or not,
such that $\Pr[Y_{j,i,B}=1]=(1-2\varepsilon)y_{j,i,B}^{*}$ and deterministically
$\sum_{B\in\tilde{\calB}}\sum_{i=1}^{m}Y_{j,i,B}\le1$.

In the remaining steps, we handle each machine $i\in[m]$ independently
and we apply to it exactly the same steps and argumentations as in
our algorithm for one machine. Intutively, this works since for each
job $j$ there can be at most one machine $i$ for which there is
a block $B\in\tilde{\calB}$ with $Y_{j,i,B}=1$.
We discard all long jobs, using the same definition for long \aw{jobs} as before. With the same proof as in the single machine case, \aw{we can show that} there can be at most $\Delta$ long jobs assigned to a block and \aw{a machine $i$} in expectation, meaning that we discard at most $\Delta \cdot m \cdot \fab{|}\aw{\tilde{\calB}} \fab{|}$ jobs in expectation. Furthermore,
 for each machine $i\in[m]$ and each block $B\in\tilde{\calB}$ we can
bound the probability for the bad event $\calE_{total}$ with the
same arguments as before. 
Also, for each job $j\in\Jl$ we can bound
the probability for the bad events $\calE_{right}(j)$ and $\calE_{left}(j)$
with the same arguments. 
Therefore, we can prove the following lemma.
\begin{lem}
\label{lem:rounding-good-m} In \fab{polynomial time} we can compute a
feasible solution to the given instance whose expected number of jobs
is at least
\[
\left((1-2\eps)\sum_{j\in\Jl}y_{j}^{*}\right)-4(\log_{2}T+1)/\eps^{4}\cdot m \cdot|\tilde{\calB}|\ge(1-O(\varepsilon))\sum_{j\in\Jl}y_{j}^{*}-O(\eps)\sum_{j\in J}y_{j}^{*}.
\]
\end{lem}

Together with the rounding algorithm from Section~\ref{subsec:poly-m-machines},
we obtain our pseudo-polynomial time $(5/4+\eps)$-approximation algorithm
\aw{on $m$ machines if}
$m\fab{=O_\eps(1)}$.
\begin{thm}
\label{thr:mainPseudopoly-1} For any constant $\eps>0$, there is
a randomized $(5/4+\eps)$-approximation algorithm for Throughput
Maximization on $m$ machines with a running time of $(nT)^{\ale{O_{\eps, m}(1)}}$.
\end{thm}

Again, due to Theorem~\ref{thm:small-m} we also obtain the same
approximation ratio \aw{in polynomial time} for arbitrary values of $m$.

\begin{cor}
\label{cor:mainpseudoPoly-m} For any constant $\eps>0$, there is
a polynomial-time randomized $(5/4+\eps)$-approximation algorithm
for Throughput Maximization on $m$ machines with a running time of
$(nT)^{O_{\eps}(1)}$.
\end{cor}

\aw{Corollary \ref{cor:mainpseudoPoly-m} yields the second claim of Theorem \ref{thr:multipleMachines} which completes its proof.}

\bibliographystyle{alpha}
\bibliography{MaxThroughput-STOC06}

@inproceedings{CCGRS14,
  author    = {V. T. Chakaravarthy and
               A. R. Choudhury and
               S. Gupta and
               S. Roy and
               Y. Sabharwal},
  title     = {Improved Algorithms for Resource Allocation under Varying Capacity},
  booktitle = {ESA},
  pages     = {222--234},
  year      = {2014}
}

@InProceedings{BCES2006,
  Title                    = {A quasi-{PTAS} for unsplittable flow on line graphs},
  Author                   = {Bansal, N. and Chakrabarti, A. and Epstein, A. and Schieber, B.},
  Booktitle                = {STOC},
  Year                     = {2006},
  Pages                    = {721--729}
}

@book{leung2004handbook,
  title={Handbook of scheduling: algorithms, models, and performance analysis},
  author={Leung, Joseph YT},
  year={2004},
  publisher={Chapman and Hall/CRC}
}

@inproceedings{BFKS09,
    author = {Bansal, N. and Friggstad, Z. and  Khandekar, R. and Salavatipour, R.},
  title = {A logarithmic approximation for unsplittable flow on line graphs},
  booktitle = {SODA},
  year      = {2009},
  pages     = {702-709}
  }

@inproceedings{AGLW14,
  author    = {A. Anagnostopoulos and
               F. Grandoni and
               S. Leonardi and
               A. Wiese},
  title     = {A Mazing 2+$\varepsilon$ Approximation for Unsplittable Flow on a Path},
  booktitle = {SODA},
  pages     = {26--41},
  year      = {2014}
}

@article{BSW14,
  title={A Constant-Factor Approximation Algorithm for Unsplittable Flow on Paths},
  author={Bonsma, P. and Schulz, J. and Wiese, A.},
  journal={SIAM Journal on Computing},
  volume={43},
  pages={767--799},
  year={2014},
  publisher={SIAM}
}

@InProceedings{BGKMW15,
author                      = {J. Batra and N. Garg and A. Kumar and T. M{\"o}mke and A. Wiese},
title                       = {New Approximation Schemes for Unsplittable Flow on a Path},
booktitle                   = {SODA},
year                        = {2015},
pages                       = {47--58}
}

@inproceedings{GMWZ17,
  author    = {F. Grandoni and
               T. M{\"{o}}mke and
               A. Wiese and
               H. Zhou},
  title     = {To Augment or Not to Augment: Solving Unsplittable Flow on a Path
               by Creating Slack},
  booktitle = {SODA},
  pages     = {2411--2422},
  year      = {2017}
}

@inproceedings{GMWZ18,
  author    = {F. Grandoni and
               T. M{\"{o}}mke and
               A. Wiese and
               H. Zhou},
  title     = {A {(5/3} + {\(\epsilon\)})-approximation for unsplittable flow on
               a path: placing small tasks into boxes},
  booktitle = {STOC},
  pages     = {607--619},
  year      = {2018}
}

@inproceedings{GMW21,
  author    = {F. Grandoni and
               T. M{\"{o}}mke and
               A. Wiese},
  title     = {Faster (1+{\(\epsilon\)})-Approximation for Unsplittable Flow on a
               Path via Resource Augmentation and Back},
  booktitle = {ESA},
  volume    = {204},
  pages     = {49:1--49:15},
  year      = {2021}
}

@inproceedings{Srinivasan2001,
author = {Srinivasan, Aravind},
title = {New approaches to covering and packing problems},
year = {2001},
isbn = {0898714907},
publisher = {Society for Industrial and Applied Mathematics},
address = {USA},
booktitle = {Proceedings of the Twelfth Annual ACM-SIAM Symposium on Discrete Algorithms},
pages = {567–576},
numpages = {10},
location = {Washington, D.C., USA},
series = {SODA '01}
}

@inproceedings{GMW22,
  author    = {F. Grandoni and
               T. M{\"{o}}mke and
               A. Wiese},
  title     = {Unsplittable Flow on a Path: The Game!},
  booktitle = {SODA},
  publisher = {{SIAM}},
  year      = {2022},
  pages     = {906--926}
}

@inproceedings{GMW22STOC,
  author    = {F. Grandoni and
               T. M{\"{o}}mke and
               A. Wiese},
  title     = {A {PTAS} for unsplittable flow on a path},
  booktitle = {{STOC}},
  pages     = {289--302},
  publisher = {{ACM}},
  year      = {2022}
}

@article{AYZ95,
  author       = {Noga Alon and
                  Raphael Yuster and
                  Uri Zwick},
  title        = {Color-Coding},
  journal      = {J. {ACM}},
  volume       = {42},
  number       = {4},
  pages        = {844--856},
  year         = {1995},
  url          = {https://doi.org/10.1145/210332.210337},
  doi          = {10.1145/210332.210337},
  bibsource    = {dblp computer science bibliography, https://dblp.org}
}

@inproceedings{dynamic-programming-framework,
  title={A dynamic programming framework for non-preemptive scheduling problems on multiple machines},
  author={Im, Sungjin and Li, Shi and Moseley, Benjamin and Torng, Eric},
  booktitle={Proceedings of the twenty-sixth annual ACM-SIAM symposium on Discrete algorithms},
  pages={1070--1086},
  year={2015},
  organization={SIAM}
}

@inproceedings{armbruster2023ptas,
  title={A PTAS for minimizing weighted flow time on a single machine},
  author={Armbruster, Alexander and Rohwedder, Lars and Wiese, Andreas},
  booktitle={Proceedings of the 55th Annual ACM Symposium on Theory of Computing},
  pages={1335--1344},
  year={2023}
}

@inproceedings{armbruster2025approximability,
  title={On the Approximability of Unsplittable Flow on a Path with Time Windows},
  author={Armbruster, Alexander and Grandoni, Fabrizio and Husi{\'c}, Edin and Tinguely, Antoine and Wiese, Andreas},
  booktitle={International Conference on Integer Programming and Combinatorial Optimization},
  pages={29--42},
  year={2025},
  organization={Springer}
}

@article{lawler1990dynamic,
  title={A dynamic programming algorithm for preemptive scheduling of a single machine to minimize the number of late jobs},
  author={Lawler, Eugene L},
  journal={Annals of Operations Research},
  volume={26},
  number={1},
  pages={125--133},
  year={1990},
  publisher={Springer}
}

@article{garey1977two,
author = {Garey, M. R. and Johnson, D. S.},
title = {Two-Processor Scheduling with Start-Times and Deadlines},
year = {1977},
issue_date = {Sep 1977},
publisher = {Society for Industrial and Applied Mathematics},
address = {USA},
volume = {6},
number = {3},
issn = {0097-5397},
url = {https://doi.org/10.1137/0206029},
doi = {10.1137/0206029},
journal = {SIAM J. Comput.},
month = sep,
pages = {416–426},
numpages = {11}
}

@article{agnetis2025fifty,
  title={Fifty years of research in scheduling--theory and applications},
  author={Agnetis, Alessandro and Billaut, Jean-Charles and Pinedo, Michael and Shabtay, Dvir},
  journal={European Journal of Operational Research},
  year={2025},
  publisher={Elsevier}
}

@article{bar2001unified,
  title={A unified approach to approximating resource allocation and scheduling},
  author={Bar-Noy, Amotz and Bar-Yehuda, Reuven and Freund, Ari and Naor, Joseph and Schieber, Baruch},
  journal={Journal of the ACM (JACM)},
  volume={48},
  number={5},
  pages={1069--1090},
  year={2001},
  publisher={ACM New York, NY, USA}
}

@article{bar2001approximating,
  title={Approximating the Throughput of Multiple Machines in Real-Time Scheduling},
  author={Bar-Noy, Amotz and Guha, Sudipto and Naor, Joseph and Schieber, Baruch},
  journal={SIAM Journal on Computing},
  volume={31},
  number={2},
  pages={331--352},
  year={2001},
  publisher={SIAM}
}

@inproceedings{BGNS99,
  author       = {Amotz Bar{-}Noy and
                  Sudipto Guha and
                  Joseph Naor and
                  Baruch Schieber},
  editor       = {Jeffrey Scott Vitter and
                  Lawrence L. Larmore and
                  Frank Thomson Leighton},
  title        = {Approximating the Throughput of Multiple Machines Under Real-Time
                  Scheduling},
  booktitle    = {Proceedings of the Thirty-First Annual {ACM} Symposium on Theory of
                  Computing, May 1-4, 1999, Atlanta, Georgia, {USA}},
  pages        = {622--631},
  publisher    = {{ACM}},
  year         = {1999},
  url          = {https://doi.org/10.1145/301250.301420},
  doi          = {10.1145/301250.301420},
  bibsource    = {dblp computer science bibliography, https://dblp.org}
}

@article{berman2000multi,
  title={Multi-phase algorithms for throughput maximization for real-time scheduling},
  author={Berman, Piotr and DasGupta, Bhaskar},
  journal={Journal of Combinatorial Optimization},
  volume={4},
  pages={307--323},
  year={2000},
  publisher={Springer}
}

@book{blazewicz2013scheduling,
  title     = {Scheduling Computer and Manufacturing Processes},
  author    = {Blazewicz, Jacek and Ecker, Klaus H. and Pesch, Erwin and Schmidt, Gunter and Weglarz, Jan},
  year      = {2013},
  publisher = {Springer Science \& Business Media},
  address   = {Berlin, Heidelberg}
}

@book{brucker2004scheduling,
  author       = {Peter Brucker},
  title        = {Scheduling algorithms},
  publisher    = {Springer},
  year         = {2004},
}

@article{chakrabarti2007approximation,
  title={Approximation algorithms for the unsplittable flow problem},
  author={Chakrabarti, Amit and Chekuri, Chandra and Gupta, Anupam and Kumar, Amit},
  journal={Algorithmica},
  volume={47},
  number={1},
  pages={53--78},
  year={2007},
  publisher={Springer}
}

@article{chekuri2001approximation,
  title={Approximation techniques for average completion time scheduling},
  author={Chekuri, Chandra and Motwani, Rajeev and Natarajan, Balas and Stein, Clifford},
  journal={SIAM Journal on Computing},
  volume={31},
  number={1},
  pages={146--166},
  year={2001},
  publisher={SIAM}
}

@article{chuzhoy2006approximation,
  title={Approximation algorithms for the job interval selection problem and related scheduling problems},
  author={Chuzhoy, Julia and Ostrovsky, Rafail and Rabani, Yuval},
  journal={Mathematics of Operations Research},
  volume={31},
  number={4},
  pages={730--738},
  year={2006},
  publisher={INFORMS}
}

@inproceedings{COR01,
  author       = {Julia Chuzhoy and
                  Rafail Ostrovsky and
                  Yuval Rabani},
  title        = {Approximation Algorithms for the Job Interval Selection Problem and
                  Related Scheduling Problems},
  booktitle    = {42nd Annual Symposium on Foundations of Computer Science, {FOCS} 2001,
                  Las Vegas, Nevada, USA, October 14-17, 2001},
  pages        = {348--356},
  publisher    = {{IEEE} Computer Society},
  year         = {2001},
  url          = {https://doi.org/10.1109/SFCS.2001.959909},
  doi          = {10.1109/SFCS.2001.959909},
  bibsource    = {dblp computer science bibliography, https://dblp.org}
}

@book{dubhashi2009concentration,
  title={Concentration of measure for the analysis of randomized algorithms},
  author={Dubhashi, Devdatt P and Panconesi, Alessandro},
  year={2009},
  publisher={Cambridge University Press}
}

@article{fischetti1987fixed,
  title     = {The Fixed Job Schedule Problem with Spread-Time Constraints},
  author    = {Fischetti, Matteo and Martello, Silvano and Toth, Paolo},
  journal   = {Operations Research},
  volume    = {35},
  number    = {6},
  pages     = {849--858},
  year      = {1987},
  publisher = {INFORMS}
}

@article{GareyJST81,
  author       = {M. R. Garey and
                  David S. Johnson and
                  Barbara B. Simons and
                  Robert Endre Tarjan},
  title        = {Scheduling Unit-Time Tasks with Arbitrary Release Times and Deadlines},
  journal      = {{SIAM} J. Comput.},
  volume       = {10},
  number       = {2},
  pages        = {256--269},
  year         = {1981},
  url          = {https://doi.org/10.1137/0210018},
  doi          = {10.1137/0210018},
  bibsource    = {dblp computer science bibliography, https://dblp.org}
}

@article{gavinsky2015tail,
  title={A tail bound for read-k families of functions},
  author={Gavinsky, Dmitry and Lovett, Shachar and Saks, Michael and Srinivasan, Srikanth},
  journal={Random Structures \& Algorithms},
  volume={47},
  number={1},
  pages={99--108},
  year={2015},
  publisher={Wiley Online Library}
}

@inproceedings{grandoni2015improved,
  title={Improved approximation algorithms for unsplittable flow on a path with time windows},
  author={Grandoni, Fabrizio and Ingala, Salvatore and Uniyal, Sumedha},
  booktitle={International Workshop on Approximation and Online Algorithms},
  pages={13--24},
  year={2015},
  organization={Springer}
}

@article{hall1994maximizing,
  title     = {Maximizing the Value of a Space Mission},
  author    = {Hall, Nicholas G. and Magazine, Michael J.},
  journal   = {European Journal of Operational Research},
  volume    = {78},
  number    = {2},
  pages     = {224--241},
  year      = {1994},
  publisher = {Elsevier}
}

@article{hochbaum1987using,
  title={Using dual approximation algorithms for scheduling problems theoretical and practical results},
  author={Hochbaum, Dorit S and Shmoys, David B},
  journal={Journal of the ACM (JACM)},
  volume={34},
  number={1},
  pages={144--162},
  year={1987},
  publisher={ACM New York, NY, USA}
}

@article{im2020breaking,
  title={Breaking 1-1/e barrier for nonpreemptive throughput maximization},
  author={Im, Sungjin and Li, Shi and Moseley, Benjamin},
  journal={SIAM Journal on Discrete Mathematics},
  volume={34},
  number={3},
  pages={1649--1669},
  year={2020},
  publisher={SIAM}
}

@inproceedings{ILM17,
  author       = {Sungjin Im and
                  Shi Li and
                  Benjamin Moseley},
  editor       = {Friedrich Eisenbrand and
                  Jochen K{\"{o}}nemann},
  title        = {Breaking 1 - 1/e Barrier for Non-preemptive Throughput Maximization},
  booktitle    = {Integer Programming and Combinatorial Optimization - 19th International
                  Conference, {IPCO} 2017, Waterloo, ON, Canada, June 26-28, 2017, Proceedings},
  series       = {Lecture Notes in Computer Science},
  volume       = {10328},
  pages        = {292--304},
  publisher    = {Springer},
  year         = {2017},
  url          = {https://doi.org/10.1007/978-3-319-59250-3\_24},
  doi          = {10.1007/978-3-319-59250-3\_24},
  bibsource    = {dblp computer science bibliography, https://dblp.org}
}

@article{jansen2010eptas,
  title={An EPTAS for scheduling jobs on uniform processors: using an MILP relaxation with a constant number of integral variables},
  author={Jansen, Klaus},
  journal={SIAM Journal on Discrete Mathematics},
  volume={24},
  number={2},
  pages={457--485},
  year={2010},
  publisher={SIAM}
}

@article{jansen2020closing,
  title={Closing the gap for makespan scheduling via sparsification techniques},
  author={Jansen, Klaus and Klein, Kim-Manuel and Verschae, Jos{\'e}},
  journal={Mathematics of Operations Research},
  volume={45},
  number={4},
  pages={1371--1392},
  year={2020},
  publisher={INFORMS}
}

@incollection{lawler1993sequencing,
  title     = {Sequencing and Scheduling: Algorithms and Complexity},
  author    = {Lawler, Eugene L. and Lenstra, Jan Karel and Rinnooy Kan, A. H. G. and Shmoys, David B.},
  booktitle = {Handbooks in Operations Research and Management Science},
  volume    = {4},
  editor    = {Graves, S. C. and Rinnooy Kan, A. H. G. and Zipkin, P. H.},
  pages     = {445--522},
  year      = {1993},
  publisher = {Elsevier},
  address   = {Amsterdam}
}

@article{lee1985simple,
  title={A simple on-line bin-packing algorithm},
  author={Lee, Chan C and Lee, Der-Tsai},
  journal={Journal of the ACM (JACM)},
  volume={32},
  number={3},
  pages={562--572},
  year={1985},
  publisher={ACM New York, NY, USA}
}

@book{pinedo08,
  title={Scheduling: Theory, Algorithms, and Systems},
  author={Pinedo, Michael L.},
  year={2008},
  publisher={Springer}
}

@article{pruhs2007approximation,
  title={Approximation schemes for a class of subset selection problems},
  author={Pruhs, Kirk and Woeginger, Gerhard J},
  journal={Theoretical Computer Science},
  volume={382},
  number={2},
  pages={151--156},
  year={2007},
  publisher={Elsevier}
}

@book{schrijver1998theory,
  title={Theory of linear and integer programming},
  author={Schrijver, Alexander},
  year={1998},
  publisher={John Wiley \& Sons}
}

@book{schrijver2003combinatorial,
  title={Combinatorial optimization: polyhedra and efficiency},
  author={Schrijver, Alexander},
  volume={24},
  year={2003},
  publisher={Springer}
}

@article{spieksma1999approximability,
  title={On the approximability of an interval scheduling problem},
  author={Spieksma, Frits CR},
  journal={Journal of Scheduling},
  volume={2},
  number={5},
  pages={215--227},
  year={1999},
  publisher={Wiley Online Library}
}

@article{DBLP:journals/siamcomp/BansalP14,
  author    = {Nikhil Bansal and
               Kirk Pruhs},
  title     = {The Geometry of Scheduling},
  journal   = {{SIAM} Journal on Computing},
  volume    = {43},
  number    = {5},
  pages     = {1684--1698},
  year      = {2014}
}

@book{williamson2011design,
  title={The design of approximation algorithms},
  author={Williamson, David P and Shmoys, David B},
  year={2011},
  publisher={Cambridge University Press}
}

@article{schmidt1990HashFunctions,
author = {Schmidt, Jeanette P. and Siegel, Alan},
title = {The Spatial Complexity of Oblivious k-Probe Hash Functions},
journal = {SIAM Journal on Computing},
volume = {19},
number = {5},
pages = {775-786},
year = {1990},
doi = {10.1137/0219054},

URL = { 
    
        https://doi.org/10.1137/0219054
    
    

}
}

@book{grotschel2012geometric,
  title={Geometric algorithms and combinatorial optimization},
  author={Gr{\"o}tschel, Martin and Lov{\'a}sz, L{\'a}szl{\'o} and Schrijver, Alexander},
  volume={2},
  year={2012},
  publisher={Springer Science \& Business Media}
}

@article{hohn2018unsplittable,
  title={How unsplittable-flow-covering helps scheduling with job-dependent cost functions},
  author={H{\"o}hn, Wiebke and Mestre, Juli{\'a}n and Wiese, Andreas},
  journal={Algorithmica},
  volume={80},
  number={4},
  pages={1191--1213},
  year={2018},
  publisher={Springer}
}

@inproceedings{antoniadis2017qptas,
  title={A QPTAS for the General Scheduling Problem with Identical Release Dates},
  author={Antoniadis, Antonios and Hoeksma, Ruben and Mei{\ss}ner, Julie and Verschae, Jos{\'e} and Wiese, Andreas},
  booktitle={44th International Colloquium on Automata, Languages, and Programming (ICALP 2017)},
  pages={31--1},
  year={2017},
  organization={Schloss Dagstuhl--Leibniz-Zentrum f{\"u}r Informatik}
}

@article{cheung2017primal,
  title={A primal-dual approximation algorithm for min-sum single-machine scheduling problems},
  author={Cheung, Maurice and Mestre, Juli{\'a}n and Shmoys, David B and Verschae, Jos{\'e}},
  journal={SIAM Journal on Discrete Mathematics},
  volume={31},
  number={2},
  pages={825--838},
  year={2017},
  publisher={SIAM}
}

\newpage
\appendix

\section{
The block construction by Chuzhoy, Ostrovsky and Rabani
}\label{sec:appendix-block}


Before we head to the proof of Lemma~\ref{lem:block-preprocessing}, we sketch the block construction by Chuzhoy et al.\ \cite{chuzhoy2006approximation}. The main subroutine $\greedy(\cdot)$ gets in input a set of jobs $J'\subseteq J$, and greedily schedules a subset $J''\subseteq J'$ according to the earliest finishing time criterion. In more detail, initially $J''=\emptyset$. At each iteration, we are given a schedule of $J''$, and we add to $J''$ a job $j\in J'\setminus J''$, if any, that can scheduled in an interval $[s(j),s(j)+p_j)\subseteq \tw(j)$ not overlapping with previously scheduled jobs, where $s(j)+p_j$ is as small as possible. We stop when no more job can be added to $J''$.%
\footnote{We remark that this is precisely the $2$-approximation algorithm in \cite{spieksma1999approximability}.}
Notice that $\greedy(\cdot)$ takes polynomial time.


    After running $\greedy(J)$ and scheduling the set of jobs $J_1 \subseteq J$, we divide $[0,T)$ into a set of blocks $\tilde\calB$ such that in each block (except possibly the last one) $\tilde K^3$ many jobs from $J_1$ are scheduled, where $\tilde K = 6/\varepsilon$.
    Starting with $i=2$, we repeat the following steps:
    Given the set of jobs $J_{i-1}$ and a set of blocks $\tilde\calB$, we 
    \ale{do the following.
    We iterate through the blocks from left to right and start with the jobs $\hat{J}=J_{i-1}$. For each block $B\in \tilde{B}$, we use $\greedy(\hat{J})$ to schedule jobs from $\hat{J}$, but only into the interval $B$. 
    If more than $\bar{K}^{i+2}$ \fab{jobs} are scheduled within $B$, we partition $B$ into smaller blocks, each containing $\bar{K}^{i+2}$
    scheduled jobs (except possibly the last) and continue with the next block and the unscheduled jobs.
    And if at most $\bar{K}^{i+2}$ \ale{jobs} are scheduled within $B$, we discard the schedule and the algorithm directly continues with the next block and all jobs in $\hat{J}$, i.e., including those just scheduled by $\greedy(\hat{J})$ in $B$.}
    We increment $i$ by one and repeat.

    The process stops after at most $6/\varepsilon \cdot \ln (6/\varepsilon) +1$ iterations (we do not describe the precise stopping condition in \cite{chuzhoy2006approximation} as it is not relevant for our analysis). 
    \ale{The output is two sets of blocks $\calB^{I}$ and $\calB^{II}$ such that $\calB^{I}\cup \calB^{II}$ is a partition of $[0, T)$, and two disjoint sets of jobs $J^{I}$ and $J_{pass}$ together with a schedule $s$ of $J^{I}$ within the blocks $B^{I}$. They obtain the following result:}
    \begin{lem}[\cite{chuzhoy2006approximation}]\label{lem:chuzhoy}
        There exists a \fab{feasible} schedule \fab{of} $\OPT''\subseteq J$ with the following properties:
        \begin{itemize}
            \item $|\OPT''| \ant{\geq} (1-\eps)|\OPT|$,
            \item each job $j\in \OPT''$ is scheduled in some block $B\in \calB^{I}\cup \calB^{II}$
            \item for each block $B\in\calB^{II}$ there are at most $K_{0}=(1/\eps)^{O(1/\eps\log(1/\eps))}$ \fab{jobs} of $\OPT''$ that are scheduled within $B$,
            \item The set \fab{of} jobs in $\OPT''$ scheduled in the blocks $\calB^{I}$ is $J^I$ and they are scheduled according to $s$.
        \end{itemize}
    \end{lem}
    
\begin{proof}[Proof of Lemma~\ref{lem:block-preprocessing}]\alnote{I added the full blocks}
    \ale{We use the algorithm from \cite{chuzhoy2006approximation} to compute the blocks $\calB^{I}$, $\calB^{II}$ and the schedule $s$ with the properties stated in Lemma~\ref{lem:chuzhoy}. Then we continue as follows. While there is a block $B\in \calB^{I}$ such that $s$ schedules more than $K_0$ jobs within $B$, then we subdivide $B$ into multiple blocks, such that within each new block (except possibly the last on) the schedule $s$ schedules exactly $K_0$ jobs. Let $\calB^{I}_+$ the resulting blocks and let $\calB_0= \calB^{I}_+\cup \calB^{II}$. This defines the block decomposition and we use the solution $\OPT''$ from Lemma~\ref{lem:chuzhoy} as the reference solution.}
    
    The properties \ref{item:cor-opt}, \ref{item:cor-restricted} and \ref{item:cor-block-bound} follow directly from Lemma~\ref{lem:chuzhoy} \ale{and the above construction.}
    To show \ref{item:cor-bound-opt}, as pointed out in the proof of \cite[Lemma~2.3]{chuzhoy2006approximation}, in every iteration, at most $(\varepsilon/6)^3|\OPT|$ blocks are added. Since \ale{there is one block in the beginning and} there are at most $(6/\varepsilon)\ln(6/\varepsilon)+1$ iterations \ale{in the algorithm from \cite{chuzhoy2006approximation} we have that }  $|\calB^{I}|+|\calB^{II}|\leq 1+ (\varepsilon/6)^3 \cdot \big((6/\varepsilon)\ln(6/\varepsilon)+1\big) |\OPT| \leq \ale{\frac{\varepsilon}{2}} |\OPT''|$ blocks in total (the last inequality used \ref{item:cor-opt} \fab{and the fact that $|\OPT|$ is large enough}).
    \ale{When we subdivide a block from $B\in \calB^{I}$ into new blocks, in each new block (except possibly the last one) the schedule $s$ schedules exactly $K_0$ jobs, showing $|\calB^{I}_+|\leq |\calB^{I}|+|\OPT|/K_0\leq |\calB^{I}|+\frac{\varepsilon}{2}|\OPT''|$. Consequently, we obtain $|\calB_0|\fab{=} |\calB^{I}_+|+|\calB^{II}|\leq |\calB^{I}|+\fab{\frac{\eps}{2}|\OPT''|}+|\calB^{II}|\leq \varepsilon |\OPT''|$, completing the proof.}
\end{proof}


\end{document}